\documentclass[journal,12pt,onecolumn,draftcls]{IEEEtran}
\IEEEoverridecommandlockouts

\usepackage{amsmath, amsthm, amssymb, amsfonts, mathtools, xargs, tensor, units, cite, stmaryrd, mathrsfs, algpseudocode, comment, graphicx}
\usepackage{multibib}
\usepackage[svgnames]{xcolor}
\usepackage[unicode=true, bookmarks=true, bookmarksnumbered=false, bookmarksopen=false, breaklinks=false, pdfborder={0 0 1}, backref=false, colorlinks=false, hidelinks]{hyperref}
\makeatletter
\let\NAT@parse\undefined
\makeatother

\newtheorem{manuallemmainner}{Lemma}
\newenvironment{manuallemma}[1]{%
  \renewcommand\themanuallemmainner{#1}%
  \manuallemmainner
}{\endmanuallemmainner}
\newtheorem{manualtheoreminner}{Theorem}
\newenvironment{manualtheorem}[1]{%
  \renewcommand\themanualtheoreminner{#1}%
  \manualtheoreminner
}{\endmanualtheoreminner}
\theoremstyle{definition}
\newtheorem{defn}{Definition}
\theoremstyle{definition}
\newtheorem{assum}{Assumption}
\newtheorem{lem}{Lemma}
\newtheorem{thm}{Theorem}

\theoremstyle{remark}

\allowdisplaybreaks
	\title{\Large Safety aware model-based reinforcement learning for optimal control of a class of output-feedback nonlinear systems}
\author{\normalsize S M Nahid Mahmud$^{1}$ \and Moad Abudia$^{1}$ \and Scott A Nivison$^{2}$ \and Zachary I. Bell$^{2}$ \and Rushikesh Kamalapurkar$^{1}$
\thanks{*This research was supported, in part, by the Air Force Research Laboratories under award number FA8651-19-2-0009. Any opinions, findings, or recommendations in this article are those of the author(s), and do not necessarily reflect the views of the sponsoring agencies.}%
\thanks{$^{1}$School of Mechanical and Aerospace Engineering, Oklahoma State University, email: {\tt\footnotesize \{nahid.mahmud, abudia@okstate.edu, rushikesh.kamalapurkar\} @okstate.edu}.}%
\thanks{$^{2}$ Air Force Research Laboratories, Florida, USA, email: {
\tt\footnotesize \{scott.nivison, zachary.bell.10\}
 @us.af.mil.}}}
\begin{document}
\maketitle
\thispagestyle{plain}
\pagestyle{plain}
\begin{abstract} 
\normalsize The ability to learn and execute optimal control policies safely is critical to realization of complex autonomy, especially where task restarts are not available and/or the systems are safety-critical. Safety requirements are often expressed in terms of state and/or control constraints. Methods such as barrier transformation and control barrier functions have been successfully used, in conjunction with model-based reinforcement learning, for safe learning in systems under state constraints, to learn the optimal control policy. However, existing barrier-based safe learning methods rely on full state feedback. In this paper, an output-feedback safe model-based reinforcement learning technique is developed that utilizes a novel dynamic state estimator to implement simultaneous learning and control for a class of safety-critical systems with partially observable state. 
\end{abstract}
\section{Introduction}
Deployment of unmanned autonomous systems in complex, dull, dirty, and dangerous tasks has seen a stark increase over the past decade due to advantages such as repeatability, accuracy, and lack of physical fatigue. To realize complex autonomy, techniques that allow autonomous agents to learn to perform tasks, in a provably safe manner, are needed. While recent years have seen prolific progress in the area of model-based safe reinforcement learning (RL) in deterministic systems
\cite{SCC.Cohen.ea2020,SCC.Yang.Vamvoudakis.ea2019,SCC.Greene.Deptula.ea2020,SCC.Mahmud.Nivison.eatoappear}, most existing techniques for model-based safe RL in deterministic systems require full state feedback. This paper focuses on the development of a model-based RL (MBRL) framework, for a class of nonlinear systems, in continuous time, under output feedback, while guaranteeing stability and safety.

RL is a trial-and-error learning process in which an agent interacts with the environment by taking an action according to a policy. The action changes the state of the agent-environment system based on the state transition model of the system. The agent receives a real-valued reward based on the current system state, the action, and the system state resulting from the action. The agent observes the rewards and adjusts its policy to maximize the cumulative reward\cite{SCC.Sutton.Barto1998,SCC.Doya2000}. While RL has been applied to solve a variety of difficult control problems, many trials are typically needed before a task is learned, and safety during the learning process is difficult to guarantee\cite{SCC.Heger.ea1994}.  Over the past few decades, numerous model-free \cite{SCC.Garcia.Fernadez.ea2015,SCC.Mihatsch.Neuneier.ea2002,SCC.Moldovan.Abbeel.ea2012,SCC.Wachi.Sui.ea2020} and model-based \cite{SCC.Meyn.Tweedie1992,SCC.Puterman2014,SCC.Cohen.ea2020,SCC.Yang.Vamvoudakis.ea2019,SCC.Greene.Deptula.ea2020,SCC.Mahmud.Nivison.eatoappear,SCC.Fisac.Akametalu.ea2018,SCC.Li.Kalabi.ea2018,SCC.Yutong.Nan.ea2021} safe RL methods have been developed for online optimal policy synthesis. Motivated by the fact that model-based methods have been shown to improve sample efficiency \cite{SCC.Lee.Park.ea2012,SCC.Modares.Lewis2014,SCC.Yang.Liu.ea2015,SCC.Wawrzynski2009,SCC.Zhang.Cui.ea2011,SCC.Adam.Busoniu.ea2012,SCC.Luo.Wu.ea2014} and have been used to relax exploration requirements such as persistence of excitation \cite{SCC.Kamalapurkar.Dinh.ea2015,SCC.Kamalapurkar.Walters.ea2016a,SCC.Kamalapurkar.Rosenfeld.ea2016,SCC.Kamalapurkar.Andrews.ea2017,SCC.Kamalapurkar.Klotz.ea2018}, this paper focuses on safe MBRL (SMBRL), where safety requirements are expressed in terms of state and/or control constraints \cite{SCC.He.Li.ea2017}.

SMBRL methods in Markov decision process (MDP) models are generally formulated in discrete time, and for systems with finite state and action spaces \cite{SCC.Howard1960,SCC.Meyn.Tweedie1992,SCC.Puterman2014}. SMBRL methods in the context of continuous-time stochastic systems are developed in  \cite{SCC.Fisac.Akametalu.ea2018,SCC.Li.Kalabi.ea2018,SCC.Yutong.Nan.ea2021}, these are offline, episodic methods that result in probabilistic safety guarantees. While probabilistic safety guarantees are common in the robotics community, applications such as manned aviation often demand deterministic safety guarantees and online, real-time learning. The deterministic safety problem  has recently been studied in the context of SMBRL for continuous-time systems in \cite{SCC.Cohen.ea2020,SCC.Yang.Vamvoudakis.ea2019,SCC.Greene.Deptula.ea2020,SCC.Mahmud.Nivison.eatoappear}. In \cite{SCC.Cohen.ea2020}, online, real-time implementation of safe RL is enabled by casting the proximity penalty approach developed in \cite{SCC.Walters.Kamalapurkar.ea2015,SCC.Deptula.Rosenfeld.ea2018,SCC.Deptula.ea2020}, into the framework of control barrier functions\cite{SCC.Ames.Xu.ea2017} (CBFs). While the CBF approach results in safety guarantees, the existence of a smooth value function, in spite of a nonsmooth cost function, needs to be assumed. Furthermore, to facilitate parametric approximation of the value function, the existence of a forward invariant compact set in the interior of the safe set needs to be established. Since the invariant set needs to be in the interior of the safe set, the CBF is not needed to establish the invariant set. As such, the penalty and the CBF, while practically useful, become theoretically superfluous.

This paper is inspired by the barrier transformation (BT) approach to state-constrained optimal control, introduced in \cite{SCC.Yang.Vamvoudakis.ea2019}. The BT approach utilizes the transformation introduced in \cite{SCC.Graichen.Petit.ea2009}, along with exact model knowledge, to transform the state constrained optimal control problem into an equivalent, unconstrained optimization problem. The unconstrained problem is then solved using adaptive optimal control methods under persistence of excitation. In \cite{SCC.Greene.Deptula.ea2020}, the results in \cite{SCC.Yang.Vamvoudakis.ea2019} are extended to soften the restrictive persistence of excitation requirement using a MBRL formulation. In \cite{SCC.Mahmud.Nivison.eatoappear}, the results in \cite{SCC.Greene.Deptula.ea2020} are extended to yield a MBRL solution to the online state-constrained optimal feedback control problem under parametric uncertainty using a filtered concurrent learning technique. While BT approaches in results such as \cite{SCC.Yang.Vamvoudakis.ea2019, SCC.Greene.Deptula.ea2020,SCC.Mahmud.Nivison.eatoappear} yield verifiable safe feedback controllers, they require full state feedback. The focus of this paper is online solution of state-constrained optimal control problems, under output feedback, for nonlinear control-affine systems in Brunovsky canonical form. While estimates of the system state can be obtained by numerically differentiating the output, to avoid noise amplification \cite{SCC.Ahnert.Abel.ea2007}, state estimators/observers are typically preferred \cite{SCC.Ghanes.Barbot.ea2017,SCC.Dinh.Kamalapurkar.ea2014,SCC.Dabroom.Khalil.ea1997,SCC.Chitour.ea2002}.
In this paper, a novel state estimator that integrates with the BT is designed to enable output feedback SMBRL (OF-SMBRL).

While partially observable systems have long been a focus of study in RL  \cite{SCC.Jaakkola.Singh.ea1995,SCC.Lewis.Vamvoudakis.ea2011,SCC.Modares.Lewis.ea2014,SCC.Omidshafiei.Pazis.ea2017}, output feedback RL results for nonlinear, continuous time systems are scarce. A method for output feedback RL in nonlinear continuous time systems in the Brunovsky canonical form is developed in the authors' preliminary work in \cite{SCC.Self.Harlan.ea2019} for \emph{unconstrained} optimal control. However, to the best of the authors' knowledge, online RL solutions to safety-constrained optimal control problems in partially observable nonlinear continuous-time systems are not available in the literature. In this paper, a novel technique is developed by incorporating a newly designed state estimator with a BT to learn control policies online, for a class of partially observable nonlinear systems, while maintaining safety and stability during the learning and execution phases. 

The following section describes the problem, and control objective.
\begin{figure}
        \centering
		\includegraphics[width=0.75\columnwidth]{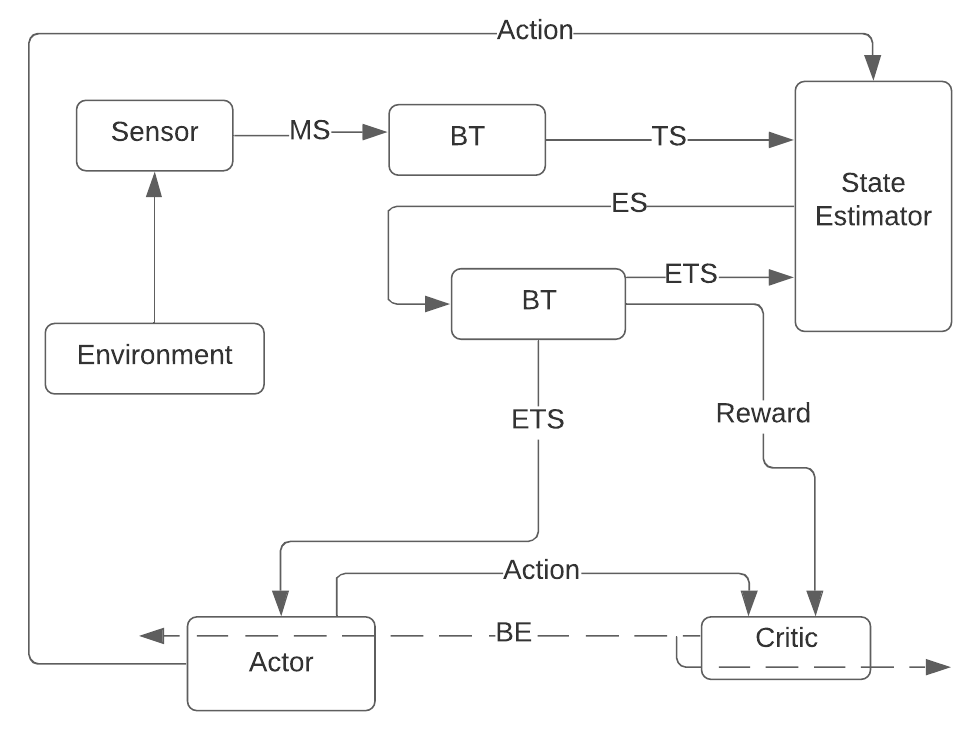}
		\caption{ Developed OF-SMBRL framework that uses a simulation-based BT-actor-critic-estimator architecture. The critic utilizes estimated transformed state variables, actions, and the corresponding estimated transformed state-derivatives to learn the value function. In the figure, BT: Barrier Transformation; MS: Measured State; TS: Transformed State; ES: Estimated State; ETS: Estimated Transformed State; BE: Bellman Error.}
		\label{fig:TNNS}
\end{figure}
\section{Problem Formulation} \label{control object}
We consider the following continuous-time affine nonlinear dynamical system in Brunovsky canonical form. 
\begin{align}
    \dot x_1 = x_2, \quad
    \dot x_2 = f(x)+g(x)u, \label{eq:Dynamics}
\end{align}
where
$x \coloneqq [x_{1};x_{2}] \in \mathbb{R}^{2n}$ is the system state, $u \in \mathbb{R}^{m}$ is the control input, and $x_{1} \in \mathbb{R}^{n}$ is the output. The drift dynamics, $f: \mathbb{R}^{2n} \rightarrow \mathbb{R}^{n}$, and control effectiveness, $g: \mathbb{R}^{2n} \rightarrow \mathbb{R}^{n \times m}$, are known, locally Lipschitz continuous functions. Let $\hat x_{1}$ and $\hat x_{2}$
be the estimates of $x_{1}$ and $x_{2}$ respectively and $\hat{x} \coloneqq [\hat{x}_{1};\hat{x}_{2}]$. The notation $[v;w]$ is used to denote the vector $[v^T \quad w^T]^{T}$, and the notation $z_{i_{j}}$ is used to denote the $j$th element of the vector $z_{i}$. The notation $I_{o}$ denotes the identity matrix of size $o$.

The objective is to design an adaptive estimator to estimate the state online, using input-output measurements, and to simultaneously estimate and utilize an output feedback optimal controller, $u$, such that starting from a given feasible initial condition $x^{0}$, the trajectories $x(\cdot)$ decay to the origin and satisfy $x_{i_{j}}(t) \in (a_{i_{j}},A_{i_{j}})$, $\forall t\geq 0$, where $ i = 1,2$, $j = 1,2, \hdots, n$, and the constants $ a_{i_{j}} < 0 < A_{i_{j}}$ define the safety constraints.

To achieve the stated objective, an OF-SMBRL framework (Fig. \ref{fig:TNNS}) is developed in this paper. A BT, inspired by \cite{SCC.Yang.Vamvoudakis.ea2019}, is used to transform the system dynamics to a barrier coordinate system. Since the transformed system is not in Brunovsky cannonical form, state estimators such as \cite{SCC.Dinh.Kamalapurkar.ea2014,SCC.Self.Harlan.ea2019} can not be utilized to estimate the transformed state variables. This paper develops a novel state estimator that leverages the Burnovsky form of the original system to yield convergence and boundedness of the transformed state estimates. The state estimator is integrated with the SMBRL framework to learn feedback control policies resulting in safe approximately optimal trajectories that satisfy the desired constraints.



In the following, Section \ref{Barrier Transformation} introduces the BT. Section \ref{sec:Velocity-estimator-designCH3} and Section \ref{Stability Analysis for the state estimation} detail the novel state estimator and an analysis of bounds on the resulting state estimation errors, respectively. Section \ref{Model-Based Reinforcement Learning} describes the novel SMBRL technique for synthesizing feedback control policies in transformed coordinates. In Section \ref{sec:Analysis}, a Lypaunov-based analysis, in the transformed coordinates, is utilized to establish practical stability of the closed-loop system resulting from the developed SMBRL technique. Guarantees that the safety requirements are satisfied in the original coordinates are also established. Simulation results in Section \ref{Simulation} compare the performance of the developed SMBRL approach to an offline pseudospectral optimal control method.

\section{Barrier Transformation}\label{Barrier Transformation}
This section introduces the barrier transformation and formalizes the connection between trajectories of the transformed system and the original system.
\begin{defn}
Given $\underline{a}<0<\overline{a}$, The function $b_{(\underline{a},\overline{a})} : \mathbb{R} \rightarrow \mathbb{R}$, referred to as the barrier function (BF), is defined as $ b_{(\underline{a},\overline{a})}(y) \coloneqq \log \left(\frac{\overline{a}(\underline{a} - y)}{\underline{a}(\overline{a} - y)}\right)$.
\end{defn}
The inverse of the BF on the interval $(\underline{a},\overline{a})$ containing the origin is given by $ b^{-1}_{(\underline{a},\overline{a})}(y) = \left(\underline{a}\overline{a}\frac{e^{y} - 1}{\underline{a}e^{y} - \overline{a}}\right) $. For $i\in\{1,2\}$ and $j\in \{1,\hdots, n\}$, the barrier transformation is a nonlinear coordinate transformation, defined as $s_{i_{j}} \coloneqq b_{(a_{i_{j}},A_{i_{j}})}(x_{i_{j}}), \quad \text{and} \quad x_{i_{j}} = b^{-1}_{(a_{i_{j}},A_{i_{j}})}(s_{i_{j}}).$
Evaluating the derivative of the inverse of the barrier function at $s_{i_{j}}$ yields $\frac{\mathrm{d}b^{-1}_{(a_{i_{j}},A_{i_{j}})}(s_{i_{j}})}{\mathrm{d}s_{i_{j}}} = \frac{1}{B_{i_j}(s_{i_{j}})}$, where $B_{i_j}(s_{i_{j}}) \coloneqq \frac{a_{i_{j}}^{2} e^{s_{i_{j}}} - 2a_{i_{j}} A_{i_{j}} + A_{i_{j}}^{2} e^{-s_{i_{j}}}}{A_{i_{j}}a_{i_{j}}^{2} - a_{i_{j}} A_{i_{j}}^{2}}.$ Let $B(s) \coloneqq [B_{1_{1}}(s_{1_{1}}); \hdots ; B_{2_{n}}(s_{2_{n}})]$.

Let $a = [a_{1};a_{2}]$, with $a_{1} = [a_{1_1}; \hdots ;a_{1_n}]$, $a_{2} = [a_{2_1}; \hdots ;a_{2_n}]$, and $A = [A_{1}; A_{2}]$, with $A_{1} = [A_{1_1}; \hdots ;A_{1_n}]$ and $A_{2} = [A_{2_1}; \hdots ;A_{2_n}]$. Let the notation $(y)_i$ be used to denote the $i$th element of any vector $y$. In the following, for any vector $y$, comprised of components of $x$, such that $y = [(x)_p;\hdots;(x)_q]$, with $1\leq p \leq q \leq 2n$, the notation $b(y)$ is used to denote componentwise application of the barrier function with the appropriate limits selected componentwise from the vectors $a$ and $A$. That is, $b(y) \coloneqq [b_{((a)_p,(A)_p)}((x)_p); \hdots ; b_{((a)_q,(A)_q)}((x)_q)]$. Similarly, given any vector $y$, comprised of components of $s$, such that $y = [(s)_p,\cdots,(s)_q]^T$, with $1\leq p \leq q \leq 2n$, $b^{-1}(y) \coloneqq [b^{-1}_{((a)_p,(A)_p)}((s)_p); \hdots ; b^{-1}_{((a)_q,(A)_q)}((s)_q)]^T$.

To transform the dynamics in \eqref{eq:Dynamics} using the BT, the time derivative of the transformed state, $s \coloneqq [s_{1};s_{2}] \in \mathbb{R}^{2n}$, can be computed as
\begin{align}\label{eq:BTDynamics}
    \dot{s}_{1} =  H(s),\quad \dot{s}_{2} = F(s) + G(s)u,
\end{align}
where $ \left(H(s)\right)_j \coloneqq B_{1_j}\left(s_{1_j}\right)b^{-1}(s_{2_{j}})$,
$\left(F(s)\right)_j \coloneqq B_{2_j} \left(s_{2_j}\right) \left(f(b^{-1}(s))\right)_j,$
and \\
$ \left(G(s)\right)_j  \coloneqq B_{2_j} \left(s_{2_j}\right) \left(g(b^{-1}(s))\right)_j.$

\subsection{Analysis of trajectories}
In the following lemma, the trajectories of the original system and the transformed system are shown to be related by the barrier transformation provided the trajectories of the transformed system are \emph{complete} \cite[page 33]{SCC.Sanfelice2021}. The completeness condition is not vacuous, it is not difficult to construct a system where the transformed trajectories escape to infinity in finite time, while the original trajectories are complete. For example, consider the system $ \dot{x} = x + x^2u $ with $x\in\mathbb{R}$ and $ u\in\mathbb{R} $. All nonzero solutions of the corresponding transformed system $ \dot{s} = B(s)b^{-1}_{(-0.5,0.5)}(s) + B(s) \left(b^{-1}_{(-0.5,0.5)}(s)\right)^2 u $ under the feedback $ \zeta(s,t) = -b^{-1}_{(-0.5,0.5)}(s) $ escape in finite time. However, all nonzero solutions of the original system under the feedback $\xi(x,t) = \zeta(b_{(-0.5,0.5)}(x),t) = -x $ converge to either $ -1 $ or $ 1 $.

\begin{lem}\label{lem:trajectoryRelation}If $t \mapsto \Phi\big(t,b(x^{0}),\zeta\big)$ is a complete Carath\'{e}odory solution to \eqref{eq:BTDynamics}, starting from the initial condition $b(x^{0})$, under the feedback policy $(s,t) \mapsto \zeta (s,t)$ and $t \mapsto \Lambda(t,x^{0},\xi)$ is a Carath\'{e}odory solution to \eqref{eq:Dynamics}, starting from the initial condition $x^{0}$, under the feedback policy $(x,t) \mapsto \xi(x,t)$, defined as $\xi(x,t) = \zeta(b(x),t)$, then $t \mapsto \Lambda(t,x^{0}, \xi)$ is complete and $ \Lambda(t,x^{0}, \xi) = b^{-1}\left(\Phi(t,b(x^{0}),\zeta)\right) $ for all $t \in \mathbb{R}_{\geq 0}$.
\end{lem}
\begin{proof}
See Appendix \ref{appendix:lemma1}.
\end{proof}

Note that the feedback $\xi$ is well-defined at $x$ only if $b(x)$ is well-defined, which is the case whenever $x$ is inside the barrier. As such, the main conclusion of the lemma also implies that $\Lambda(\cdot,x^0,\xi)$ remains inside the barrier. It is thus inferred from Lemma \ref{lem:trajectoryRelation} that if the trajectories of \eqref{eq:BTDynamics} are bounded and decay to a neighborhood of the origin under a feedback policy $(s,t) \mapsto \zeta (s,t)$, then the feedback policy $(x,t) \mapsto \zeta \big(b(x),t \big)$, when applied to the original system in \eqref{eq:Dynamics}, achieves the control objective stated in Section \ref{control object}.

In this paper, a feedback controller is designed, using output feedback, in the transformed coordinates and leveraging Lemma \ref{lem:trajectoryRelation}, a feedback controller that keeps the trajectories of the original system within the safe bounds, without using full state measurements, is generated. Note that the unknown part of the state, $x_{2}$, is simply the time derivative of the output, $x_{1}$. While the derivative can be computed numerically, state estimators, such as the one designed in the following section, have been shown to be more robust to measurement noise than numerical differentiation \cite{SCC.Ghanes.Barbot.ea2017,SCC.Dinh.Kamalapurkar.ea2014,SCC.Dabroom.Khalil.ea1997,SCC.Chitour.ea2002}.
Furthermore, the state estimator designed in the following section allows for rigorous inclusion of state estimation errors in the analysis of the controller. 

\section{State Estimation}\label{sec:Velocity-estimator-designCH3}
In this section, a state estimator
inspired by \cite{SCC.Self.Harlan.ea2019} is developed to generate estimates of $x$. The state estimator relies on the Brunovsky form in \eqref{eq:Dynamics}. Since the transformed system in \eqref{eq:BTDynamics} is not in the Brunovsky form, the state estimator is derived in the original coordinates with a feedback term that is derived using Lyapunov-based analysis of the estimation error in the transformed coordinates. The estimator is given by
\begin{align}\label{eq:baseEqx}
    \dot{\hat{x}}_{1} = \hat{x}_{2}, \quad \dot{\hat{x}}_{2} = f\left(\hat{x}\right)+g\left(\hat{x}\right)u+\nu_{1},
\end{align}
where, $\nu_{1} = [\nu_{1_1};\hdots;\nu_{1_n}] \in \mathbb{R}^n$ is a feedback term designed in the following. The design of $\nu_{1}$ is motivated by the need to establish bounds (precisely, \eqref{eq:need_nu1_design} of the appendix) on the state estimation errors in the transformed coordinates. To facilitate the design of $\nu_{1}$, let the state estimation errors be defined as $ \tilde{x}_{1}=x_{1}-\hat{x}_{1}$, and $ \tilde{x}_{2} = x_{2}-\hat{x}_{2}$. The feedback component $\nu_{1_{j}}$ where $j\in \{1,\hdots, n\}$ is designed as 
\begin{equation}
    \nu_{1_{j}}= \frac{\alpha^{2}(b(x_{1_{j}})-b(\hat{x}_{1_{j}}))-\left(k+\alpha+\beta_{1}\right)\eta_{{j}}}{B_{1_j}\left(b\left(\hat{x}_{1_j}\right)\right)}, \label{eq:5}
\end{equation}
where the signal $\eta_{{j}}$ is added to compensate for the fact that $x_{2_{j}}$ is not measurable. Based on
the subsequent stability analysis, the signal $\eta_{{j}}$ is designed
as the output of the dynamic filter 
\begin{equation}
	\dot{\eta_{j}}=-\beta_{1}\eta_{j}-kr_{j}-\alpha \left(\frac{\mathrm{d}}{\mathrm{d}t} \left(b
	(x_{1_{j}})-b(\hat{x}_{1_{j}})\right)\right),\label{eq:6}
\end{equation}
where $\eta_{j}\left(t_{0}\right)=0$, $\alpha$, $k$, and $\beta_{1}$ are positive constants, $t_{0}$ is the initial time and the error signal $r_{{j}}$ is defined as
\begin{equation}
r_{j}= \frac{\mathrm{d}}{\mathrm{d}t} 
\left(b(x_{1_{j}})-b(\hat{x}_{1_{j}})\right)+\alpha (b(x_{1_{j}})-b(\hat{x}_{1_{j}}))+\eta_{j}.\label{eq:7}
\end{equation}
The signal $\eta_{j}$ can be implemented, without numerically differentiating $b(x_{1_j})$, via the filter,
\begin{align}
    \dot{\overline{\eta_{j}}} &= -(k+\beta_{1})\eta_{j}-k\alpha\Big(b(x_{1_{j}})-b(\hat{x}_{1_{j}})\Big),\dot{\overline{\eta_{j}}}(0)=0,\nonumber\\
    \eta_{j}(t) &= \overline{\eta_{j}}(t)- (k+\alpha)\Big(b(x_{1_{j}}(t))-b(\hat{x}_{1_{j}}(t))  -b(x_{1_{j}})(0)+b(\hat{x}_{1_{j}})(0) \Big),\label{eq:8}
\end{align}
where, $\overline{\eta_{j}}$ is an auxiliary signal.
To facilitate the analysis, which is done in transformed coordinates, an equivalent expression of the state estimator in the transformed coordinates is needed. To transform the state estimator using the BT, let $ \hat{s}_{i_{j}} \coloneqq b(\hat{x}_{i_{j}})$, and $\tilde{s}_{i_{j}} \coloneqq s_{i_{j}} - \hat{s}_{i_{j}}$. The state estimator can then be expressed in transformed coordinates $\hat{s} \coloneqq [\hat{s}_{1};\hat{s}_{2}] \in \mathbb{R}^{2n}$, as
\begin{align}\label{eq:BTDynamicsestimated}
     \dot{\hat{s}}_{1} = H(\hat{s}), \quad \dot{\hat{s}}_{2} = F(\hat{s}) + G(\hat{s})u + \nu_{2},
\end{align}
where $\nu_{2} = [\nu_{2_1};\hdots;\nu_{2_n}] \in \mathbb{R}^n$ is given by $ \nu_{2_{j}}=B_{2_j}\left(\hat{s}_{2_j}\right)\nu_{1_{j}}$, where $j\in \{1,\hdots, n\}$.

As detailed in Lemma \ref{lem2} below, the design of the BT ensures that the trajectories of \eqref{eq:Dynamics}, \eqref{eq:baseEqx}, \eqref{eq:5}, \eqref{eq:6}, \eqref{eq:7}, and \eqref{eq:BTDynamicsestimated} are linked by the BT whenever the underlying state trajectories $x(\cdot)$ and $s(\cdot)$ and the initial conditions $\hat{x}^{0}$ and $\hat{s}^{0}$ are linked by the BT. 

\begin{lem}\label{lem2}If $t \mapsto \Psi\big(t;b(x_{1}(\cdot)),b(\hat{x}^{0}) \big)$ is a Carath\'{e}odory solution to \eqref{eq:BTDynamicsestimated} along the trajectory $x_{1}(\cdot)$, starting from the initial condition $b(\hat{x}^{0})$, and if $t \mapsto \xi(t;x_{1}(\cdot),\hat{x}^{0})$ is a Carath\'{e}odory solution to \eqref{eq:baseEqx}, starting from the initial condition $\hat{x}^{0}$, along the trajectory $x_{1}(\cdot)$, then $ \xi(t;x_{1}(\cdot),\hat{x}^{0}) = b^{-1}\big(\Psi\big(t;b(x_{1}(\cdot)),b(\hat{x}^{0})\big)\big) $ for all $t \in \mathbb{R}_{\geq 0}$.
\end{lem}
\begin{proof}
See Appendix \ref{appendix:lemma2}.
\end{proof}
The following section develops a bound on a Lyapunov-like function of the state estimation errors to be utilized in the subsequent stability analysis.
\section{Errors bounds for the state estimator}\label{Stability Analysis for the state estimation}
To develop bounds on the estimation errors, consider the time-derivative of \eqref{eq:BTDynamics}: 
     $ \ddot{{s}}_{1}  = F_{2}(s)  + F_{3}(s) + G_{1}(s)u$,
where 
$(F_{3}(s))_j  \coloneqq B_{1_j}\left(s_{1_j}\right)f\left([b^{-1}(s_{1_{j}}), b^{-1}(s_{2_{j}})]\right)$, \\ $ (G_{1}(s))_j\coloneqq B_{1_j}\left(s_{1_j}\right) g\left([b^{-1}(s_{1_{j}}), b^{-1}(s_{2_{j}})]\right)$, and 
     $(F_{2}(s))_j  \coloneqq \left(\frac{a_{1_{j}}^{2}e^{s_{1_{j}}}-A_{1_{j}}^{2}e^{-s_{1_{j}}}} {A_{1_{j}}a_{1_{j}}^{2} - a_{1_{j}}A_{1_{j}}^{2}}\right)  b^{-1}(s_{2_{j}}).$
Similarly, the time-derivative of \eqref{eq:BTDynamicsestimated} yields
      $\ddot{\hat{s}}_{1} = F_{2}(\hat{s})  + F_{3}(\hat{s}) + G_{1}(\hat{s})u + \nu_{3},$
where $\nu_{3} = [\nu_{3_1};\hdots;\nu_{3_n}] \in \mathbb{R}^n$ and
    $\nu_{3_{j}}=\beta_{1_j}\left(\hat{s}_{1_j}\right)\nu_{1_{j}}$, which yields
$\nu_{3_{j}} = (\alpha^{2}\tilde{s}_{1_{j}}-\left(k+\alpha+\beta\right)\eta_{j}),$ and \eqref{eq:6} can be expressed as
\begin{equation}\label{eq:eta_final}
	\dot{\eta_{j}} =-\beta_{1}\eta_{j}-kr_{j}-\alpha\dot{\tilde{s}}_{1_{j}}.
\end{equation}
In vector form, \eqref{eq:eta_final} can be expressed as
	$\dot{\eta} =-\beta_{1}\eta-kr-\alpha(\tilde{H}(s,\hat{s})),$
where $\tilde{H}(s,\hat{s}) \coloneqq H(s) - H(\hat{s}) = \dot{\tilde{s}}_1$, $\eta = [\eta_{1};\hdots;\eta_{n}]$, and $r = [r_{1};\hdots;r_{n}]$.
Furthermore, \eqref{eq:7} can be rewritten as $ r = \dot{\tilde{s}}_{1} + \alpha\tilde{s}_{1} + \eta$, which yields
\begin{multline}
    \dot r = F_{2}({s})+ F_{3}({s}) + G_{1}({s})\hat{u}(\hat{s},\hat{W}_{a}) -  F_{2}(\hat{s}) - F_{3}(\hat{s}) \\- G_{1}(\hat{s})\hat{u}(\hat{s},\hat{W}_{a}) - \alpha^2 \tilde{s}_{1} + (k+\alpha+\beta_{1})\eta \nonumber+ \alpha\dot{\tilde{s}}_{1} -\beta_{1}\eta-kr-\alpha\dot{\tilde{s}}_{1},
\end{multline}
which can be expressed as
\begin{equation}
\label{eq:r_totalfinal}
    \dot r = \tilde{F_2}(s,\hat{s}) + \tilde{F_{3}}(s,\hat{s}) + \tilde{G}_{1}(s,\hat{s})\hat{u}(\hat{s},\hat{W}_{a}) - \alpha^2 \tilde{s}_{1} - kr + k\eta + \alpha\eta,
\end{equation}
where $\tilde{F_2}(s,\hat{s}) \coloneqq F_{2}(s) - F_{2}(\hat{s})$, $\tilde{F_3}(s,\hat{s}) \coloneqq F_{3}(s) - F_{3}(\hat{s})$, $\tilde{G_1}(s,\hat{s}) \coloneqq G_{1}(s) - G_{1}(\hat{s})$. 
The following lemma develops a bound on a Lyapunov-like function of the state estimation errors $\tilde{s}_{1}$, $r$, and $\eta$. The bound is utilized in the subsequent stability 
analysis in Section $\ref{sec:Analysis}$.
\begin{lem}\label{lem3}
Let $V_{se} : \mathbb{R}^{3n} \rightarrow \mathbb{R}_{\geq 0}$ be a continuously differentiable candidate Lyapunov function defined as
$ V_{se}(Z_{1}) \coloneqq \frac{\alpha^2}{2}\tilde{s}_{1}^T\tilde{s}_{1} + \frac{1}{2}r^Tr + \frac{1}{2}\eta^T\eta$,
where $Z_{1} \coloneqq [\tilde{s}_{1}^T,r^T,\eta^T]$. Provided $s$, $\hat{s} \in \overline{B}(0,\chi)$, where $\overline{B}(0,\chi)$ is the closed ball of radius $\chi>0$ centered at the origin, the orbital derivative of $V_{se}$, along the trajectories of $\dot{\tilde{s}}_{1}$, $\dot{r}$, and $\dot{\eta}$, defined as $\dot{V}_{se}(Z_{1},s,\tilde{s},\tilde{W}_{a}) \coloneqq \frac{\partial V_{se}(Z_{1},s,\tilde{s},\tilde{W}_{a})}{\partial \tilde{s}_{1}}(H(s)-H(\hat{s})) + \frac{\partial V_{se}(Z_{1},s,\tilde{s},\tilde{W}_{a})}{\partial r}\dot{r} + \frac{\partial V_{se}(Z_{1},s,\tilde{s},\tilde{W}_{a})}{\partial \eta}\dot{\eta}$, can be bounded as $\dot{V}_{se}(Z_{1},s,\tilde{s},\tilde{W}_{a}) \leq -{\alpha^3}\|\tilde{s}_{1}\|^{2} - (k-\varpi_{1}\varpi_{4})\|r\|^{2} - (\beta_{1}-\alpha)\|\eta\|^{2} + \varpi_{1}\left(1+\varpi_{4}+\varpi_{4}\alpha\|\right) \|r\|\|\tilde{s}_{1}\|  + \varpi_{1}\varpi_{4} \|r\| \|\eta\|  \\ + \varpi_{2} \|r\| \|\tilde{W}_{a}\|  +  \varpi_{3} \|r\|$.
\end{lem}
\begin{proof}
See Appendix \ref{appendix:lemma3}.
\end{proof}
The following section develops a novel OF-SMBRL technique for synthesizing  feedback  control  policies.
\section{Safe Model-based Reinforcement Learning}\label{Model-Based Reinforcement Learning}
Lemma \ref{lem:trajectoryRelation} implies that if a feedback controller that practically stabilizes the transformed system in \eqref{eq:BTDynamics} is designed, then the same feedback controller, applied to the original system by inverting the BT also achieves the control objective stated in Section \ref{control object}. In the following, a controller that practically stabilizes \eqref{eq:BTDynamics} is designed as an estimate of a controller that minimizes the infinite horizon cost\footnote{A state penalty function $x\mapsto E(x)$, given in the original coordinates, can easily be transformed into an equivalent state penalty $Q(s) = E(b^{-1}(s))$. Since the barrier function is monotonic and $b(0) = 0$, if $E$ is positive definite, then so is $Q$. Furthermore, for applications with bounded control inputs, a non-quadratic penalty function similar to Eq. 17 of \cite{SCC.Yang.Ding.ea2020} can be incorporated in \eqref{cost function}.}
\begin{equation} \label{cost function}
J(u(\cdot)) \coloneqq 	\int_{0}^\infty c(\phi(\tau,s^0,u(\cdot)), u(\tau)) d\tau,
\end{equation}
over the set $\mathcal{U}$ of piecewise continuous functions $t\mapsto u(t)$, 
subject to \eqref{eq:BTDynamics},
where $\phi(\tau, s^0, u (\cdot))$ denotes the trajectory of
(\ref{eq:BTDynamics}), evaluated at time $\tau$, starting from the state $s^0$, and
under the controller $u (\cdot)$. In \eqref{cost function}, $c(s,u) \coloneqq Q'(s) + u^{T}Ru$, with $Q'(s) \coloneqq [q_1((s)_1),\cdots, q_{2n}((s)_{2n})]$, where $q_i: \mathbb{R} \mapsto \mathbb{R} $ are strictly monotonic functions with $q_i(0) = 0$, for all $i=1,...,2n$, and $R \in \mathbb{R}^{m \times m}$ is a symmetric positive definite (PD) matrix. For the optimal value function to be a Lyapunov function for the optimal policy, the following assumption is needed \cite{SCC.Self.Harlan.ea2019}.

\begin{assum}\label{ass:CostRestrictions} One of the following is true:\begin{enumerate}
	\item $ Q' $ is \text{PD}.
	\item $ Q' $ is \text{positive semidefinite (PSD)}, and $s_{1}\mapsto Q'\left(s\right) $ is \text{PD} for all nonzero $ s_{2}\in \mathbb{R}^{n}$.
	\item $ Q' $ is \text{PSD}, $s_{2}\mapsto Q'\left(s\right) $ is \text{PD} for all nonzero $ s_{1}\in \mathbb{R}^{n} $ and $ F\left(s\right)\neq 0 $ whenever $ s_{1} \neq 0 $.
\end{enumerate}\end{assum}
Assuming that an optimal controller exists, let the optimal value function, denoted by $V^{*} : \mathbb{R}^{n} \times \mathbb{R}^q  \rightarrow \mathbb{R} $, be defined as
\begin{equation}V^{*}(s) := \min_{u(\cdot)\in \mathcal{U}_{[t,\infty})}\int_{t}^\infty c(\phi(\tau,s,u_{[0,\tau)}(\cdot)), u(\cdot)) d\tau, \label{eq:valuefunction}\end{equation}
where $u_I$ and $\mathcal{U}_I$ are obtained by restricting the domains of $u$ and functions in $\mathcal{U}_I$ to the interval $ I \subseteq \mathbb{R} $, respectively. Assuming that the optimal value function is continuously differentiable, it can be shown to be the unique PD solution of the Hamilton-Jacobi-Bellman (HJB) equation \cite[Theorem 1.5]{SCC.Kamalapurkar.Walters.ea2018}
\begin{equation}\label{HJB} 
\min_{u\in\mathbb{R}^q} \Big(V_{s_{1}}\left(H(s)\right)+ V_{s_{2}} \left(F(s)+G(s)u\right) + Q^\prime(s)+ u^{T}Ru\Big) = 0,
\end{equation}
where $\nabla_{\left(\cdot\right)} \coloneqq  \frac{\partial}{\partial \left(\cdot\right)}$, and $V_{\left(\cdot\right)} \coloneqq  \nabla_{\left(\cdot\right)} V$. Furthermore, the optimal controller is given by the feedback policy $u(t) = u^*(\phi(t,s,u_{[0,t)}))$ where $ u^{*}: \mathbb{R}^{n} \rightarrow \mathbb{R}^{m} $ defined as
\begin{equation}\label{eq:optimalcontrol}
    u^{*}(s) := -\frac{1}{2}R^{-1}G(s)^{T}(\nabla_{s_{2}}V^{*}(s))^{T}.
\end{equation}
\subsection{Stability Under Optimal State Feedback\label{sec:Stability Under Optimal Feedback}}
The following theorem establishes global asymptotic stability of the closed-loop system under optimal state feedback.
\begin{thm}\label{thm:Optimal GAS}
	If the optimal state feedback controller \eqref{eq:optimalcontrol} that minimizes the cost function in \eqref{cost function} exists and if the corresponding optimal value function is continuously differentiable and radially unbounded, then the origin of closed-loop system $ \dot{s}_{1} = H(s) $, and $ \dot{s}_{2} = F({s}) + G({s})u^{*}(s) $ is globally asymptotically stable.
\end{thm}
\begin{proof}
See Appendix \ref{appendix:thm:Optimal GAS}.
\end{proof}

\subsection{Value function approximation}
Since computation of analytical solutions of the HJB equation is generally infeasible, especially for nonlinear systems, parametric approximation methods are used to approximate the value function $V^{*}$, and the optimal policy $u^{*}$. The optimal value function is expressed as
\begin{equation} \label{eq:optimalV}
    V^{*}\left(s\right)=W^{T}\sigma\left(s\right)+\epsilon\left(s\right),
\end{equation}
where $W\in\mathbb{R}^{L}$ is an unknown vector of bounded weights, $\sigma:\mathbb{R}^{2n}\rightarrow\mathbb{R}^{L}$ is a vector of continuously differentiable nonlinear activation functions \cite[Def. 2.1]{SCC.Sadegh1993} such that $\sigma\left(0\right)=0$ and $\nabla_{s} \sigma \left(0\right)=0$, $L\in\mathbb{N}$ is the number of basis functions, and $\epsilon:\mathbb{R}^{2n}\rightarrow\mathbb{R}$ is the reconstruction error.
Exploiting the universal function approximation property \cite[Property 2.3]{SCC.Kamalapurkar.Walters.ea2018} of single layer neural networks, it can be concluded that given any compact set\footnote{Note that at this stage, the existence of a compact forward-invariant set that contains trajectories of \eqref{eq:BTDynamics} is not being assumed. The existence of such a set is established in section \ref{sec:Stability Under Optimal Feedback}, theorem \ref{thm:Optimal GAS}.} $\overline{B}\left(0,\chi\right) \subset\mathbb{R}^{2n}$ and a positive constant $\overline{\epsilon}\in\mathbb{R}$, there exists a number of basis functions $L\in\mathbb{N}$, and known positive constants $\bar{W}$ and $\overline{\sigma}$ such that $\left\Vert W\right\Vert \leq\bar{W}$, $\sup_{s\in\overline{B}\left(0,\chi\right)}\left \| \epsilon \left(s\right)\right\| \leq\overline{\epsilon}$, $\sup_{s\in\overline{B}\left(0,\chi\right)}\left\|\nabla_{s}\epsilon\left(s\right)\right\| \leq\overline{\epsilon}$, $\sup_{s\in\overline{B}\left(0,\chi\right)}\left \| \sigma \left(s\right)\right\| \leq\overline{\sigma}$, and $\sup_{s\in\overline{B}\left(0,\chi\right)}\left\|\nabla_{s}\sigma\left(s\right)\right\| \leq\overline{\sigma}$. 

Using (\ref{HJB}), a representation of the optimal controller using the same basis as the optimal value function is derived as
\begin{equation}\label{eq:optimalu}
    \medmuskip = 0mu
    \thickmuskip = 0mu
    u^{*}\left(s\right)=-\frac{1}{2}R^{-1}G^{T}\left(s\right)\left(\nabla_{s_{2}}\sigma^{T}\left(s\right)W+\nabla_{s_{2}}\epsilon^{T}\left(s\right)\right).
\end{equation}
Since the ideal weights, $W$, are unknown, an actor-critic approach is used in the following to estimate $W$. To that end, let the NN estimates $\hat{V}:\mathbb{R}^{n}\times\mathbb{R}^{L}\to\mathbb{R}$ and $\hat{u}:\mathbb{R}^{n}\times\mathbb{R}^{L}\to\mathbb{R}^{m}$ be defined as
\begin{gather}
\hat{V}\left(\hat{s},\hat{W}_{c}\right)\coloneqq\hat{W}_{c}^{T}\sigma\left(\hat{s}\right),\label{V_app}\\
\hat{u}\left(\hat{s},\hat{W}_{a}\right)\coloneqq-\frac{1}{2}R^{-1}G^{T}\left(\hat{s}\right)\nabla_{\hat{s}_{2}}\sigma^{T}\left(\hat{s}\right)\hat{W}_{a},\label{u_app}
\end{gather}
where the critic weights, $\hat{W}_{c}\in\mathbb{R}^{L}$ and actor weights, $\hat{W}_{a}\in\mathbb{R}^{L}$ are estimates of the ideal weights, $W$.
\subsection{Bellman Error}
Substituting (\ref{V_app}) and (\ref{u_app}) into (\ref{HJB}) results in a residual term, $\hat{\delta}: \mathbb{R}^{2n} \times \mathbb{R}^{L} \times \mathbb{R}^{L} \rightarrow \mathbb{R}$, referred to as the Bellman error (BE), defined as \begin{multline} \label{BE1}
    \hat{\delta}(\hat{s},\hat{W}_{c},\hat{W}_{a}) \coloneqq  \hat{V}_{\hat{s}_{1}}(\hat{s},\hat{W}_{c})\left(H(\hat{s})\right) + \hat{u}(\hat{s},\hat{W}_{a})^{T}R\hat{u}(\hat{s},\hat{W}_{a}) \\ 
    + \hat{V}_{\hat{s}_{2}}(\hat{s},\hat{W}_{c}) \left(F(\hat{s}) + G(\hat{s})\hat{u}(\hat{s},\hat{W}_{a})\right) + Q^\prime(\hat{s}).
\end{multline}
Traditionally, online RL methods require a persistence of excitation (PE) condition to be able learn the approximate control policy \cite{SCC.Modares.Lewis.ea2013,SCC.Kamalapurkar.Rosenfeld.ea2016,SCC.Kiumarsi.Lewis.ea2014}. Guaranteeing PE a priori and verifying PE online are both typically impossible. However, using virtual excitation facilitated by the model, stability and convergence of online RL can established under a PE-like condition that, while impossible to guarantee a priori, can be verified online (by monitoring the minimum eigenvalue of a matrix in the subsequent Assumption \ref{ass:CLBCADPLearnCond})\cite{SCC.Kamalapurkar.Walters.ea2016}.
Using the system model, the BE can be evaluated at any arbitrary point in the state space. Virtual excitation can then be implemented by selecting a set of state variables $\left\{ z^{k} = [z^k_1,z^k_2]^T\mid k=1,\cdots,N\right\} $, where $z^k_i\in\mathbb{R}^n$ for $i=1,2$, and evaluating the BE at this set of state variables to yield
\begin{multline} \label{BE2}
    \hat{\delta}_{k}(z^{k},\hat{W}_{c},\hat{W}_{a}) \coloneqq \hat{V}_{z^k_1}(z^{k},\hat{W}_{c})\left(H(z^{k})\right)+ Q^\prime(z^{k})\\ + \hat{V}_{z^k_2}(z^{k},\hat{W}_{c})\left(F(z^{k}) + G(z^{k})\hat{u}(z^{k},\hat{W}_{a})\right)\\
    + \hat{u}(z^{k},\hat{W}_{a})^{T}R\hat{u}(z^{k},\hat{W}_{a}).
\end{multline}
Defining the actor and critic weight estimation errors as $\tilde{W}_{c} \coloneqq W -\hat{W}_{c}$ and  $\tilde{W}_{a} \coloneqq W -\hat{W}_{a}$ and substituting the estimates \eqref{eq:optimalV} and \eqref{eq:optimalu} into (\ref{HJB}), and subtracting from \eqref{BE1}, the BE that can be expressed in terms of the weight estimation errors as\footnote{The dependence of various functions on the state, $s$, is omitted hereafter for brevity whenever it is clear from the context.}
\begin{equation} \label{Analytical BE}
\hat{\delta}=-\omega^{T}\tilde{W}_{c}+\frac{1}{4}\tilde{W}_{a}^{T}G_{\sigma}\tilde{W}_{a}+\Delta,
\end{equation}
where $\Delta\coloneqq\frac{1}{2}W^{T}\nabla_{\hat{s}_{2}} \sigma G_{R}\nabla_{\hat{s}_{2}} \epsilon^{T}+\frac{1}{4}G_{\epsilon} - \nabla_{\hat{s}_{1}}  \epsilon H-\nabla_{\hat{s}_{2}}  \epsilon F$,   
 $G_{R}\coloneqq GR^{-1}G^{T}$, $G_{\epsilon}\coloneqq \nabla_{\hat{s}_{2}} \epsilon  G_{R} \nabla_{\hat{s}_{2}}  \epsilon^{T}$, $G_{\sigma}\coloneqq \nabla_{\hat{s}_{2}}  \sigma G R^{-1}G^{T} \nabla_{\hat{s}_{2}}  \sigma^{T} $,  and  $\omega \coloneqq \nabla_{\hat{s}_{1}}  \sigma H + \nabla_{\hat{s}_{2}}  \sigma \left(F+G\hat{u}\left(\hat{s},\hat{W}_a\right)\right)$.

Similarly, (\ref{BE2}) implies that
\begin{equation} \label{Approximate BE}
    \hat{\delta}_{k}=-\omega_{k}^{T}\tilde{W}_{c}+\frac{1}{4}\tilde{W}_{a}^{T}G_{\sigma_{k}}\tilde{W}_{a}+\Delta_{k},
\end{equation}
where,
$\Delta_{k} \coloneqq \frac{1}{2}W^{T} \nabla_{z^k_2} \sigma_{k} G_{R_{k}} \nabla_{z^k_2} \epsilon_{k}^{T}+\frac{1}{4}G_{\epsilon_{k}}-\nabla_{z^k_1}\epsilon_{k} H_{k}-\nabla_{z^k_2}\epsilon_{k} F_{k}$, $G_{\epsilon_{k}}\coloneqq \nabla_{z^k_2} \epsilon_{k} G_{R_{k}} \nabla_{z^k_2} \epsilon_{k}^{T}$, $\omega_{k} \coloneqq \nabla_{z^k_1} \sigma_{k}H_{k}+\nabla_{z^k_2} \sigma_{k}\left(F_k+G_{k}\hat{u}\left(\hat{z}^k,\hat{W}_a\right)\right)$, $G_{\sigma_{k}} \coloneqq \nabla_{z^k_2} \sigma_{k} G_{k} R^{-1} G_{k}^{T} \nabla_{z^k_2} \sigma_{k}^{T}$, $G_{R_{k}}\coloneqq G_{k}R^{-1}G_{k}^{T}$, $F_{k} \coloneqq F(
z^{k})$, $G_{k} \coloneqq G(
z^{k})$, $H_{k} \coloneqq H(
z^{k})$, $\sigma_{k} \coloneqq \sigma (z^{k})$, and $\epsilon_{k} \coloneqq \epsilon(z^{k})$.

Note that $\sup_{s\in\overline{B}\left(0,\chi\right)}\left |\Delta \right| \leq d \overline{\epsilon}$ and if $z^{k} \in \overline{B}\left(0,\chi\right)$ then $ \left |\Delta_{k} \right| \leq d \overline{\epsilon}_{k}$, for some constant $d > 0$.

While the extrapolation state variables $z^k$ are assumed to be constant in this analysis for clarity, the approach extends in a straightforward manner to time-varying extrapolation state variables confined to a compact neighborhood of the origin.

\subsection{Update laws for Actor and Critic weights}
Using the extrapolated BEs $\hat{\delta}_{k}$ from (\ref{BE2}), the weights are updated according to
\begin{align}
    \dot{\hat{W}}_{c} &=- \frac{k_{c}}{N}\Gamma\sum_{k=1}^{N}\frac{\omega_{k}}{\rho_{k}}\hat\delta_{k},\label{W_c}\\
    \dot{\Gamma} &= \beta\Gamma- \frac{k_{c}}{N}\Gamma\sum_{k=1}^{N}\frac{\omega_{k}\omega_{k}^{T}}{\rho_{k}^{2}}\Gamma,\label{gamma}\\
    \dot{\hat{W}}_{a} \!\!&=\!\! -k_{a_{1}}\!\left(\!\hat{W}_{a}\!-\!\hat{W}_{c}\!\right)\!\!+\!\!\sum_{k=1}^{N}\!\!\frac{k_{c}G_{\sigma_{k}}^{T}\hat{W}_{a}\omega_{k}^{T}}{4N\rho_{k}}\hat{W}_{c}\!-\!k_{a_{2}}\!\hat{W}_{a},\label{W_a}
\end{align}
with $\Gamma\left(t_{0}\right)=\Gamma_{0}$, where $\Gamma:\mathbb{R}_{\geq t_{0}} \to \mathbb{R}^{L\times L}$
is a time-varying least-squares gain matrix, $\rho_{k}\left(t\right)\coloneqq 1+\gamma\omega_{k}^{T}\left(t\right)\omega_{k}\left(t\right)$, $\gamma > 0$ is a constant positive normalization gain,  $\beta > 0 \in \mathbb{R}$ is a constant forgetting factor, and $k_{c},k_{a_{1}},k_{a_{2}} > 0 \in \mathbb{R}$ are constant adaptation gains. 
The control commands sent to the system are then computed using the actor weights as 
\begin{equation}\label{eq:ucontrol}
    u(t)= \hat{u}\left(\hat{s}(t),\hat{W}_{a}(t)\right), \quad t\geq 0.
\end{equation}
The Lyapunov function needed to analyze the closed loop system defined by \eqref{eq:BTDynamics}, \eqref{eq:baseEqx}, \eqref{eq:5}, \eqref{eq:8},  \eqref{W_c}, \eqref{gamma}, and \eqref{W_a}  is constructed using stability properties of \eqref{eq:BTDynamics} under the optimal feedback \eqref{eq:optimalcontrol}. To that end, the following section analyzes the optimal closed-loop system.

Using Theorem \ref{thm:Optimal GAS} and the converse Lyapunov theorem for asymptotic stability \cite[Theorem 4.17]{SCC.Khalil2002}, the existence of a radially unbounded PD function $ \mathcal{V}:\mathbb{R}^{2n}\to\mathbb{R} $ and a PD function $ W:\mathbb{R}^{2n}\to\mathbb{R} $ is guaranteed such that
\begin{equation}\label{eq:Converse Lyapunov Function}
\medmuskip=0mu
\thickmuskip=0mu
\thinmuskip=0mu
    \mathcal{V}_{s_{1}}\left(s\right)F(s) +\mathcal{V}_{s_{2}}\left(s\right)\left(F\left(s\right)+G\left(s\right)u^{*}\left(s\right)\right)\leq-W\left(s\right),
\end{equation}
for all $ s\in\mathbb{R}^{2n} $. The functions $ \mathcal{V} $ and $ W $ are utilized in the following section to analyze the stability of the output feedback approximate optimal controller. 	
\section{Stability Analysis}\label{sec:Analysis}
The following PE-like rank condition is utilized in the stability analysis. 
\begin{assum}
    \label{ass:CLBCADPLearnCond}There exists a constant $\underline{c}_{1} > 0$ such that the set of points $\left\{ z^{k}\in\mathbb{R}^{n}\mid k=1,\hdots,N\right\} $ satisfies
    \begin{equation}
    \underline{c}_{1}I_{L} \leq\inf_{t\in\mathbb{R}_{\geq T}}\left(\frac{1}{N}\sum_{k=1}^{N}\frac{\omega_{k}\left(t\right)\omega_{k}^{T}\left(t\right)}{\rho_{k}^{2}\left(t\right)}\right).\label{eq:CLBCPE2}
    \end{equation}
\end{assum}
Since $\omega_{k}$ is a function of the estimates $\hat{s}$ and $\hat{W}_{a}$,  Assumption \ref{ass:CLBCADPLearnCond} cannot be guaranteed a priori. However, unlike the PE condition, Assumption \ref{ass:CLBCADPLearnCond} can be verified online. Furthermore, since $\lambda_{\min}\left(\sum_{k=1}^{N}\frac{\omega_{k}\left(t\right)\omega_{k}^{T}\left(t\right)}{\rho_{k}^{2}\left(t\right)}\right)$ is non-decreasing in the number of samples, $N$, Assumption \ref{ass:CLBCADPLearnCond} can be met, heuristically, by increasing the number of samples. It is established in \cite[Lemma 1]{SCC.Kamalapurkar.Rosenfeld.ea2016} that under 
Assumption \ref{ass:CLBCADPLearnCond} and provided $\lambda_{\min}\left\{ \Gamma_{0}^{-1}\right\} >0$, the update law in (\ref{gamma}) ensures that the least squares gain matrix satisfies 
\begin{align}\label{eq:OFBADP1Gammabound}
	\underline{\Gamma}I_{L}\leq\Gamma\left(t\right)\leq\overline{\Gamma}I_{L},		\end{align}
$\forall t\in\mathbb{R}_{\geq 0}$ and for some  $\overline{\Gamma},\underline{\Gamma}>0$. Using \eqref{eq:BTDynamics}, the orbital derivative of the PD function $\mathcal{V}$ introduced in \eqref{eq:Converse Lyapunov Function}, along the trajectories of \eqref{eq:BTDynamics}, under the controller $u= \hat{u}\left(\hat{s},\hat{W}_{a}\right)$, is given by $ \dot{\mathcal{V}}\left(s,\tilde{s},\tilde{W}_{a}\right) = \mathcal{V}_{s_{2}}\left(s\right)\left(F\left(s\right)+G\left(s\right)\hat{u}\left(\hat{s},\hat{W}_{a}\right)\right) +\mathcal{V}_{s_{1}}\left(s\right)H(s)$, where $ \tilde{s} \coloneqq s-\hat{s}$.

Using \eqref{eq:Converse Lyapunov Function} and the facts that $ G $ is bounded, the basis functions $ \sigma $ are bounded, and the value function approximation error $ \epsilon $ and its derivative with respect to $ s, \hat{s} $ are bounded on compact sets, the time-derivative can be bounded as
\begin{equation} \label{eq:Converse Lyapunov}
\dot{\mathcal{V}}\left(s,\tilde{s},\tilde{W}_{a}\right)\leq-W\left(s\right)+\iota_{1}\overline{\epsilon}+\iota_{2}\left\Vert \tilde{s}\right\Vert \left\Vert \tilde{W}_{a}\right\Vert +\iota_{3}\left\Vert \tilde{W}_{a}\right\Vert +\iota_{4}\left\Vert \tilde{s}\right\Vert,
\end{equation}for all $ \hat{W}_{a} \in \mathbb{R}^{L}$, for all $ s \in \overline{B}(0,\chi) $, and for all $ \hat{s} \in \overline{B}(0,\chi) $, where $ \iota_{1},\cdots,\iota_{4} $ are positive constants.

Let $ \Theta\left(\tilde{W}_{c},\tilde{W}_{a},t\right)\coloneqq \frac{1}{2}\tilde{W}_{c}^{T}\Gamma^{-1}\left(t\right)\tilde{W}_{c}+\frac{1}{2}\tilde{W}_{a}^{T}\tilde{W}_{a}$.
 The orbital derivative of $ \Theta $ along the trajectories of \eqref{W_c} - \eqref{W_a} is given by
 \begin{equation} \label{theta}
\dot{\Theta}\left(\tilde{W}_{c},\tilde{W}_{a},t\right)= \tilde{W}_{c}^{T}\Gamma^{-1}\dot{\tilde{W}}_{c}-\frac{1}{2}\tilde{W}_{c}^{T}\Gamma^{-1}\dot{\Gamma}\Gamma^{-1}\tilde{W}_{c}\\+\tilde{W}_{a}^{T}\dot{\tilde{W}}_{a},
\end{equation}
where $\dot{\tilde{W}}_{c}$ = $ -\dot{\hat{W}}_{c}$, and $\dot{\tilde{W}}_{a}$ =  $-\dot{\hat{W}}_{a}$. \\
Provided the extrapolation state variables are selected such that $z^{k} \in \overline{B}(0,\chi)$, $\forall k = 1,\hdots,N$, the orbital derivative in \eqref{theta} can be bounded as 
\begin{multline}\label{eq:thetadot}
       \dot{\Theta}\left(\tilde{W}_{c},\tilde{W}_{a},t\right) \leq -k_{c}\underline{c}\left\Vert \tilde{W}_{c}\right\Vert ^{2}-\left(k_{a1}+k_{a2}\right)\left\Vert \tilde{W}_{a}\right\Vert ^{2}\\+k_{c}\iota_{8}\overline{\epsilon}\left\Vert \tilde{W}_{c}\right\Vert +k_{c}\iota_{5}\left\Vert \tilde{W}_{a}\right\Vert ^{2}+\left(k_{c}\iota_{6}+k_{a1}\right)\left\Vert \tilde{W}_{c}\right\Vert \left\Vert \tilde{W}_{a}\right\Vert \\+\left(k_{c}\iota_{7}+k_{a2}\overline{W}\right)\left\Vert \tilde{W}_{a}\right\Vert, 
\end{multline}
for all $ t\geq0 $, where $ \iota_{5},\hdots,\iota_{8} $ are positive constants that are independent of the learning gains, $ \overline{W} $ denotes an upper bound on the norm of the ideal weights $ W $, and  \\ $ \underline{c_{3}} =  \inf_{t \geq 0} \lambda_{\min}\left\{\left(\frac{\beta}{2k_{c}}\Gamma^{-1}\left(t\right)+\frac{1}{2N}\sum_{k=1}^{N}\frac{\omega_{k}\omega_{k}^{T}}{\rho_{k}}\right)\right\}$. Assumption \ref{ass:CLBCADPLearnCond} and \eqref{eq:OFBADP1Gammabound} guarantee that $ \underline{c_{3}}>0 $.
From \eqref{eq:derivative_Lyapunov_function2e} we get, 
\begin{multline}\label{eq:derivative_Lyapunov_function2f}
\dot{V}_{se}\left(Z_{1},s,\tilde{s},\tilde{W}_{a}\right) \leq -{\alpha^3}\|\tilde{s}_{1}\|^{2} - (k-\varpi_{1}\varpi_{4})\|r\|^{2}   -(\beta_{1}-\alpha)\|\eta\|^{2} + \varpi_{1}\left(1+\varpi_{4}+\varpi_{4}\alpha\right) \|r\|\|\tilde{s}_{1}\|  \\+ \varpi_{1}\varpi_{4} \|r\| \|\eta\|  + \varpi_{2} \|r\| \|\tilde{W}_{a}\| +  \varpi_{3} \|r\|, 
\end{multline}
for all $ \hat{W}_{a} \in \mathbb{R}^{L}$, for all $ s \in \overline{B}(0,\chi)$, and for all $\hat{s} \in \overline{B}(0,\chi) $, where $ \varpi_{2},\varpi_{3} $ are positive constants that are independent of the learning gains and $\varpi_{1}$, $\varpi_{4}$ are the Lipschitz constants on $\overline{B}(0,\chi)$, for $F$ and $h$, respectively. 
\begin{thm}
    Provided Assumption \ref{ass:CostRestrictions}, Assumption \ref{ass:CLBCADPLearnCond}, the hypothesis of Lemma \ref{lem3}, and the hypothesis of Theorem \ref{thm:Optimal GAS} hold, the gains are selected large enough to ensure that \eqref{gain_condition} holds and the matrix $M+M^{T}$, is PD, and the weights $\hat{W}_{c}$, $\Gamma$, and $\hat{W}_{a}$ are updated according to \eqref{W_c}, \eqref{gamma}, and \eqref{W_a}, respectively, then the estimation errors $\tilde{W}_{c}$, $\tilde{W}_{a}$, and the trajectories of the transformed system in \eqref{eq:BTDynamics}, under the controller in \eqref{eq:ucontrol}, are locally uniformly ultimately bounded.
\end{thm}
\begin{proof}
The candidate Lyapunov function for the closed-loop system is selected as
\begin{equation}\label{eq:candidateLyapunovfunction}
V_{L}\left(Z,t\right)\coloneqq \mathcal{V}\left(s\right)+\Theta\left(\tilde{W}_{c},\tilde{W}_{a},t\right)+V_{se}\left(Z_{1}\right), 
\end{equation}
where $ Z\coloneqq\begin{bmatrix}
s^{T}&\tilde{s}_{1}^{T}&r^{T}&\eta^{T}&\tilde{W}_{c}^{T}&\tilde{W}_{a}^{T}
\end{bmatrix}^{T}.$ 

Let $\mathcal{C} \subset \mathbb{R}^{5n}$ be a compact set defined as
\[
    \mathcal{C} \coloneqq \left\{(s,\tilde{s}_{1},\eta,r) \in \mathbb{R}^{5n} \mid \begin{gathered}\|s\| + \|\tilde{s}_{1}\|(1+\varpi_{4}(1+\alpha)) + \varpi_{4}(\|r\|  \|+ \|\eta\|) \leq \chi\end{gathered}\right\}.
\]
Using \eqref{eq:need_nu1_design}, whenever, $(s,\tilde{s}_{1},\eta,r) \in \mathcal{C}$, it can be concluded that $s,\hat{s} \in \overline{B}(0,\chi)$. As a result, \eqref{eq:Converse Lyapunov}, \eqref{eq:thetadot}, and \eqref{eq:derivative_Lyapunov_function2f} imply that whenever $Z \in \mathcal{C} \times \mathbb{R}^{2L}$, the orbital derivative of the candidate Lyapunov function along the trajectories of \eqref{eq:BTDynamics}, \eqref{eq:baseEqx}, \eqref{eq:7}, \eqref{eq:8}, \eqref{W_c}, \eqref{gamma}, \eqref{W_a}, under the controller \eqref{eq:ucontrol}, can be bounded as
 \begin{align*}
\dot{V}_{L}\left(Z,t\right) \leq -W\left(s\right)+\iota_{1}\overline{\epsilon}+\iota_{2}\left\Vert \tilde{s}\right\Vert \left\Vert \tilde{W}_{a}\right\Vert +\iota_{3}\left\Vert \tilde{W}_{a}\right\Vert  -k_{c}\underline{c}\left\Vert \tilde{W}_{c}\right\Vert ^{2}-\left(k_{a1}+k_{a2}\right)\left\Vert \tilde{W}_{a}\right\Vert ^{2}\nonumber\\ +k_{c}\iota_{8}\overline{\epsilon}\left\Vert \tilde{W}_{c}\right\Vert +\iota_{4}\left\Vert \tilde{s}\right\Vert  +k_{c}\iota_{5}\left\Vert \tilde{W}_{a}\right\Vert ^{2}+\left(k_{c}\iota_{6}+k_{a1}\right)\left\Vert \tilde{W}_{c}\right\Vert \left\Vert \tilde{W}_{a}\right\Vert +\left(k_{c}\iota_{7}+k_{a2}\overline{W}\right)\left\Vert \tilde{W}_{a}\right\Vert\nonumber\\   -{\alpha^3}\|\tilde{s}_{1}\|^{2} - (k-\varpi_{1}\varpi_{4})\|r\|^{2}  - (\beta_{1}-\alpha)\|\eta\|^{2} + \varpi_{1} \|r\|\|\tilde{s}_{1}\|  +\varpi_{1}\varpi_{4} \|r\| \|\tilde{s}_{1}\| \nonumber\\ +\varpi_{1}\varpi_{4}\alpha\| \|r\| \tilde{s}_{1}\| + \varpi_{1}\varpi_{4} \|r\| \|\eta\|   + \varpi_{2} \|r\| \|\tilde{W}_{a}\| +  \varpi_{3} \|r\|,
\end{align*}
which can be re-expressed as,
\begin{align*}
\dot{V}_{L}\left(Z,t\right) \leq -W\left(s\right)-k_{c}\underline{c}_{3}\left\Vert \tilde{W}_{c}\right\Vert ^{2} -\left(k_{a1}+k_{a2}-k_{c}\iota_{5}\right)\left\Vert \tilde{W}_{a}\right\Vert ^{2} -{\alpha^3} \left\Vert \tilde{s}_{1} \right\Vert ^{2} - (k-\varpi_{1}\varpi_{4})\|r\|^{2} \nonumber\\ - (\beta_{1}-\alpha)\|\eta\|^{2}
 +\left(k_{c}\iota_{6}+k_{a1}\right)\left\Vert \tilde{W}_{c}\right\Vert \left\Vert \tilde{W}_{a}\right\Vert+\iota_{2}(1 + \varpi_{4} + \varpi_{4}\alpha)\|\tilde{s}_{1}\|\left\Vert \tilde{W}_{a}\right\Vert\nonumber\\ + \bigg(\iota_{2}\varpi_{4}+\varpi_{2}\bigg)\|r\|\left\Vert \tilde{W}_{a}\right\Vert + \iota_{2}\varpi_{4}\|\eta\|\left\Vert \tilde{W}_{a}\right\Vert
+ (1+\varpi_{4}+\varpi_{4}\alpha)\varpi_{1} \|r\|\|\tilde{s}_{1}\|+ \varpi_{1}\varpi_{4} \|r\| \|\eta\|\nonumber\\ + \iota_{4}\varpi_{4}\|\eta\|+ (\varpi_{3}+\iota_{4}\varpi_{4})\|r\|   +\bigg(\iota_{3}+k_{c}\iota_{7}+k_{a2}\overline{W}\bigg)\left\Vert \tilde{W}_{a}\right\Vert +k_{c}\iota_{8}\overline{\epsilon}\left\Vert \tilde{W}_{c}\right\Vert \nonumber\\ +\iota_{4}(1 + \varpi_{4} + \varpi_{4}\alpha)\|\tilde{s}_{1}\|+\iota_{1}\overline{\epsilon},
\end{align*}
which can be re-expressed as,
\begin{multline*}
\dot{V}_{L}\left(Z,t\right) \leq -W\left(s\right)-k_{c}\underline{c}_{3}\left\Vert \tilde{W}_{c}\right\Vert ^{2} -{\alpha^3} \left\Vert \tilde{s}_{1} \right\Vert ^{2} - (\beta_{1}-\alpha)\|\eta\|^{2} 
- (k-\varpi_{1}\varpi_{4})\|r\|^{2}\\ -\left(k_{a1}+k_{a2}-k_{c}\iota_{5}\right)\left\Vert \tilde{W}_{a}\right\Vert ^{2}
+\left(k_{c}\iota_{6}+k_{a1}\right)\left\Vert \tilde{W}_{c}\right\Vert \left\Vert \tilde{W}_{a}\right\Vert + \iota_{2}(1 + \varpi_{4} + \varpi_{4}\alpha)\|\tilde{s}_{1}\|\left\Vert \tilde{W}_{a}\right\Vert\\
+ \bigg(\iota_{2}\varpi_{4}+\varpi_{2}\bigg)\|r\|\left\Vert \tilde{W}_{a}\right\Vert + \iota_{2}\varpi_{4}\|\eta\|\left\Vert \tilde{W}_{a}\right\Vert + \iota_{4}\varpi_{4}\|\eta\|
+ (1+\varpi_{4}+\varpi_{4}\alpha)\varpi_{1} \|r\|\|\tilde{s}_{1}\|+ \varpi_{1}\varpi_{4} \|r\| \|\eta\|\\
+ (\varpi_{3}+\iota_{4}\varpi_{4})\|r\| +\bigg(\iota_{3}+k_{c}\iota_{7}+k_{a2}\overline{W}\bigg)\left\Vert \tilde{W}_{a}\right\Vert +k_{c}\iota_{8}\overline{\epsilon}\left\Vert \tilde{W}_{c}\right\Vert  +\iota_{4}(1 + \varpi_{4} + \varpi_{4}\alpha)\|\tilde{s}_{1}\|+\iota_{1}\overline{\epsilon},
\end{multline*}
which yields
\begin{equation*}
\dot{V}_{L}\left(Z,t\right)\leq-W\left(s\right)-z^{T}\left(\frac{M+M^{T}}{2}\right)z+Pz+\iota_{1}\overline{\epsilon},
\end{equation*}where $ z\coloneqq \begin{bmatrix}\left\Vert \tilde{W}_{c}\right\Vert  & \left\Vert \tilde{W}_{a}\right\Vert  & \left\Vert \tilde{s}_{1}\right\Vert  & \left\Vert r\right\Vert  & \left\Vert \eta\right\Vert \end{bmatrix}^{T}$, 
and the matrices $P \coloneqq \begin{bmatrix}
     k_{c}\iota_{8}\overline{\epsilon}\\ \left(k_{c}\iota_{7}+\iota_{3}+k_{a2}\overline{W}\right) \\ \iota_{4}\left(1+\varpi_{4}+\varpi_{4}\alpha\right) \\ \left(\varpi_{3}+\iota_{4}\varpi_{4}\right) \\ \iota_{4}\varpi_{4}
\end{bmatrix}^{T}$ and 

$M \coloneqq \begin{bmatrix}
    k_{c}\underline{c}_{3}&0& 0&0& 0 \\ 
    -\left(k_{c}\iota_{6}+k_{a1}\right)&\left(k_{a1}+k_{a2}-k_{c}\iota_{5}\right)& 0&0&0\\
    0&-\iota_{2}\left(1+\varpi_{4}+\varpi_{4}\alpha\right)&{\alpha^3}&0&0\\
    0&-\left(\iota_{2}\varpi_{4}+\varpi_{2}\right)&-\varpi_{1}(1+\varpi_{4}+\varpi_{4}\alpha)&\left(k- \varpi_{1}\varpi_{4}\right)&0\\
    0&-\iota_{2}\varpi_{4}& 0& -\varpi_{1}\varpi_{4}& \left(\beta_{1}-\alpha\right)
    \end{bmatrix}^{T}$.
Provided the matrix $ M+M^{T}$  is PD,
\begin{equation*}
\dot{V}_{L}\left(Z,t\right)\leq-W\left(s\right)-\underline{M}\left\Vert z\right\Vert ^{2}+\overline{P}\left\Vert z\right\Vert +\iota_{1}\overline{\epsilon},
\end{equation*}where $ \underline{M} \coloneqq \lambda_{\min}\left\{\frac{M+M^{T}}{2}\right\}$, $\overline{P}= \|P\|_{\infty}$. Letting $ \underline{M}\eqqcolon\underline{M}_{1}+\underline{M}_{2}$, and letting $\mathcal{W}:\mathbb{R}^{5n+2L} \to \mathbb{R} $ be defined as $ \mathcal{W}\left(Z\right)=-W\left(s\right)-\underline{M}_{1}\left\Vert z\right\Vert^{2}$, the orbital derivative can be bounded as 
\begin{equation}
\dot{V}_{L}\left(Z,t\right)\leq-\mathcal{W}\left(Z\right),\label{eq:OFBADPVDotBound}
\end{equation}
$\forall \left\Vert Z\right\Vert >\frac{1}{2}\left(\frac{\overline{P}}{\underline{M}_{2}}+\sqrt{\frac{\overline{P}^{2}}{\underline{M}_{2}^{2}}+\frac{\iota_{1}^{2}\overline{\epsilon}^{2}}{\underline{M}_{2}^{2}}}\right)\eqqcolon \mu, \forall Z\in\overline{B}\left(0,\bar {\chi}\right)$, for all $ t\geq 0 $, and some $\bar{\chi}$ such that $\bar B(0,\bar{\chi}) \subseteq \mathcal{C} \times \mathbb{R}^{2L}$. 

Using the bound in \eqref{eq:OFBADP1Gammabound} and the fact that the converse Lyapunov function is guaranteed to be time-independent, radially unbounded, and PD, \cite[Lemma 4.3]{SCC.Khalil2002} can be invoked to conclude that \begin{equation}
\underline{v}\left(\left\Vert Z\right\Vert \right)\leq V_{L}\left(Z,t\right)\leq\overline{v}\left(\left\Vert Z\right\Vert \right),\label{eq:OFBADPVBound}
\end{equation}
for all $t \in \mathbb{R}_{\geq 0}$ and for all $Z\in\mathbb{R}^{5n+2L}$, where $\underline{v},\overline{v}:\mathbb{R}_{\geq 0}\rightarrow\mathbb{R}_{\geq 0}$ are class $\mathcal{K}$ functions.

Provided the learning gains, the domain radii $ \chi \ \text{and} \ \bar{\chi} $, and the basis functions for function approximation are selected such that $ M+M^{T} $ is PD and 
\begin{equation} \label{gain_condition}
    \mu<\overline{v}^{-1}\left(\underline{v}\left(0,\bar {\chi}\right)\right),
\end{equation}
then \cite[Theorem 4.18]{SCC.Khalil2002} can be invoked to conclude that Z is locally uniformly ultimately bounded. Since the estimates $\hat{W}_{a}$ approximate the ideal weights $ W $, the policy $ \hat{u} $ approximates the optimal policy $ u^{*} $.
\end{proof}
\section{Simulation}\label{Simulation}
To demonstrate the performance of the developed method, two simulations are provided. One for a two-state dynamical system and one for a four-state  dynamical  system corresponding to a two-link planar robot manipulator. 
\subsection{Two state dynamical system}\label{simsec1_SE}
The dynamical system is given by \eqref{eq:Dynamics}, 
where \begin{align} \label{sim_dyn}
f(x) = -x_{1}-\frac{1}{2}x_{2}\left(1-\left(\cos\left(2x_{1}\right)+2\right)^2\right), \quad
g(x) = \cos\left(2x_{1}\right)+2,
\end{align}
and $x_{1}$ is the measured output. 
The state $x$ = $[x_{1};x_{2}]$ needs to satisfy the constraints 
$x_{1} \in (a_{1},A_{1})$ and $x_{2} \in (a_{2},A_{2})$ where $a_{1}$ = -7, $A_{1}$ = 5,  $a_{2}$ = -5, $A_{2}$ = 7.
The objective is to synthesize the policy to minimize the infinite horizon cost in (\ref{cost function}), with $Q'(s) = s^{T}Qs$ where $Q = 10I_{2}$ and $R = 0.1$. The basis functions for value function approximation are selected as $\sigma(\hat{s}) = 
[\hat{s}_{1}^{2} ; \hat{s}_{1}\hat{s}_{2} ; \hat{s}_{2}^{2}]$. The initial conditions for the state, the estimated state, and the initial guesses for the weights are selected as $
    x(0) =  [-6;6],
    \hat{x}(0) =  [-6;4]$, $\Gamma(0)= I_{3}$, and $\hat{W}_{a}(0) = \hat{W}_{c}(0) = \left[10;\nicefrac{1}{2};\nicefrac{1}{2}\right]$ respectively.
The ideal values of the actor and the critic weights for the barrier-transformed optimal control problem are unknown. The simulation uses 100 fixed Bellman error extrapolation points in a 4$\times$4  square around the origin of the $s-$coordinate system.

\subsubsection{Results for the two state system} \label{ Result_sim1}

Fig.\ref{fig:original_state_SE_sim1} indicates that the system state, $x$, stays within the user-specified safe set while converging to the origin. As  seen  from Fig. \ref{fig: Estimation error_SE_sim1}, the state estimation errors also converge to the zero. The results in Fig. \ref{fig: Weight_EE_SE_sim1} show that the unknown weights for both the actor and critic converge to similar values. \\
As the ideal actor and critic weights are unknown, the estimates cannot be directly compared against the ideal weights. To gauge the quality of the estimates, the trajectory generated by the controller
\[
    u(t) = \hat{u}\left(\hat{s}(t),\hat{W}_{c}^{*}\right),
\]
\begin{table}
    \centering
    \caption{ Comparison of costs for a single trajectory of barrier transformed state variables\eqref{eq:Dynamics}, obtained using the optimal feedback controller generated via the developed method, and obtained using pseudospectral numerical optimal control software.}
\label{table:SE_sim1}
\begin{tabular}{p{6cm}p{2cm}} 
  \hline 
  Method & Cost \\
  \hline
  OF-SMBRL   & 55.82\\
  GPOPS II \cite{SCC.GPOPS}& 55.17 \\
  \hline
\end{tabular}
\end{table}
where $\hat{W}_{c}^{*}$ is the final value of the critic weights obtained in Fig. \ref{fig: Weight_EE_SE_sim1}, starting from a specific initial condition, is compared against the trajectory obtained using an \emph{offline} numerical solution computed using the GPOPS II optimization software \cite{SCC.GPOPS}. The total cost, generated by numerically integrating \eqref{cost function}, is used as the metric for comparison. The results in Table \ref{table:SE_sim1} indicate that the two solution techniques generate slightly different trajectories in the phase space (see Fig. \ref{fig: Comparison_Optimal_SE}). Moreover, as expected, the total cost of the trajectories from OF-SMBRL method is higher than the total cost of the trajectories from GPOPS. We hypothesize that the difference in costs can be attributed to the basis for value function approximation being unknown.

\begin{figure}
        \centering
		\includegraphics[width=0.75\columnwidth]{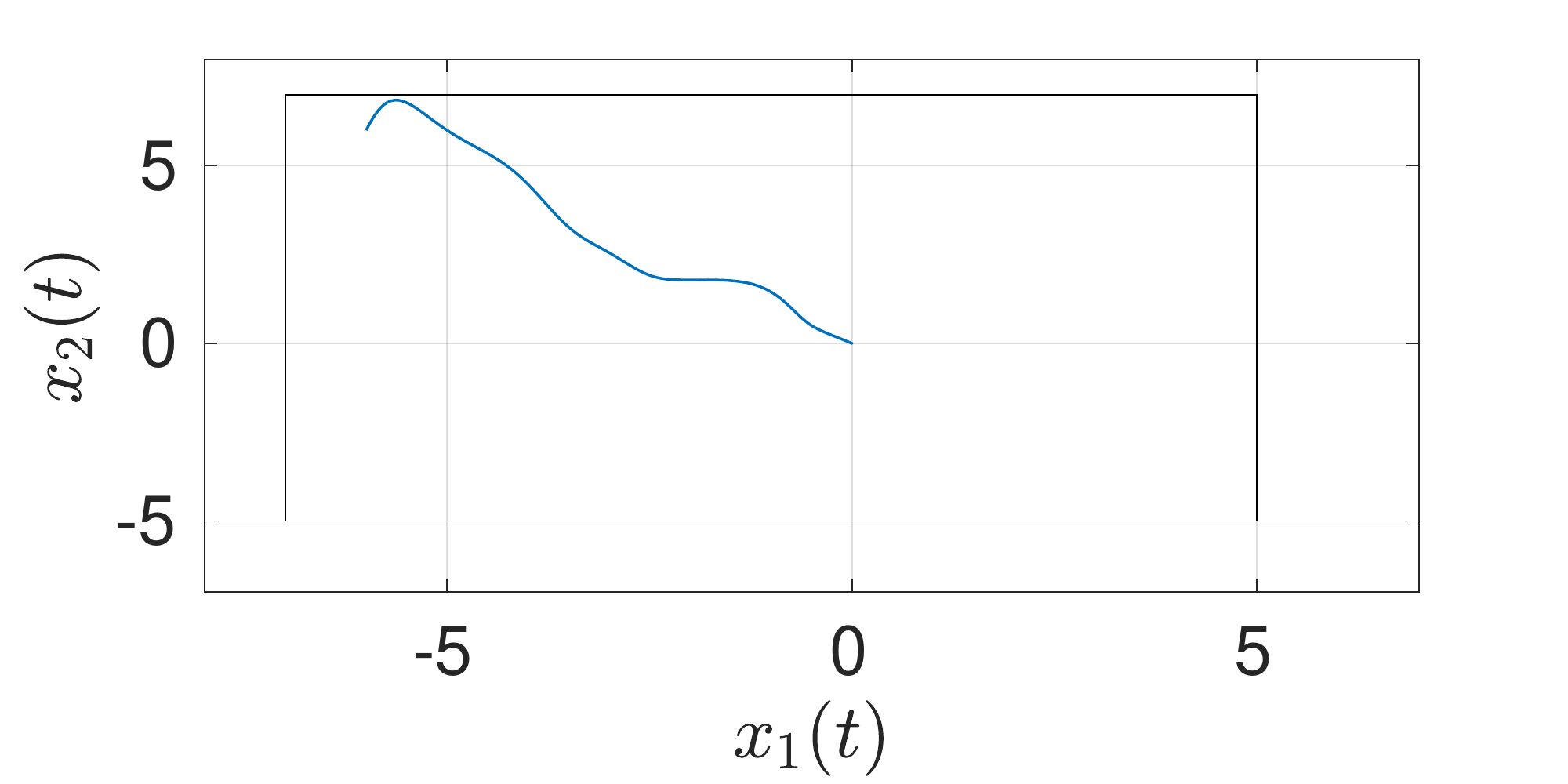}
		\caption{ Phase portrait for the two-state dynamical system using OF-SMBRL in the original coordinates. The boxed area represents the user-selected safe set.}
		\label{fig:original_state_SE_sim1}
\end{figure}

\begin{figure}
		\includegraphics[width=0.75\columnwidth]{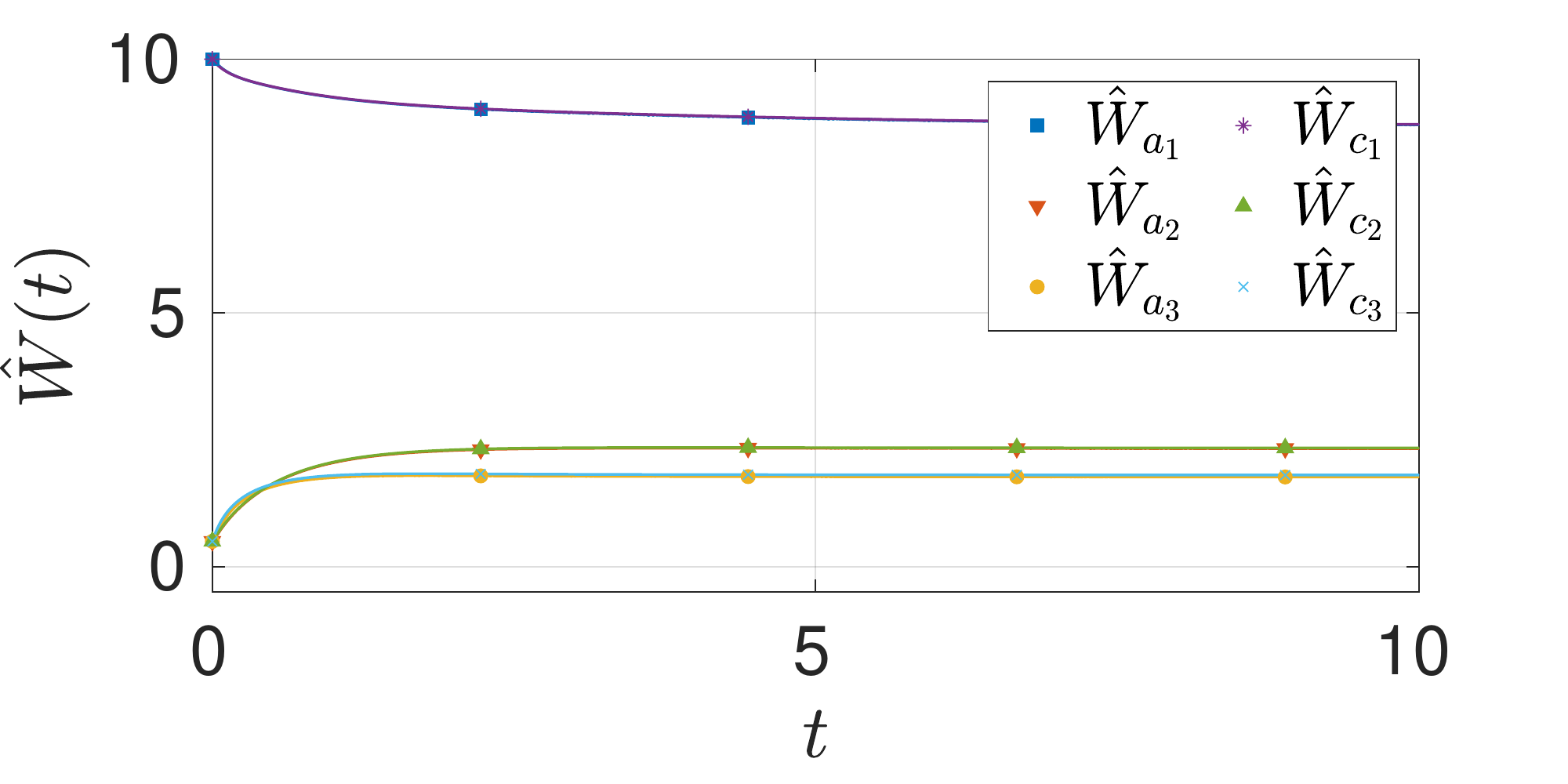}
		\caption{ Estimates of the actor and the critic weights under nominal gains from Table \ref{table:Sentivity_analysis_SE_sim1} for the two-state dynamical system.}
		\label{fig: Weight_EE_SE_sim1}
\end{figure}


\begin{figure}
		\includegraphics[width=0.75\columnwidth]{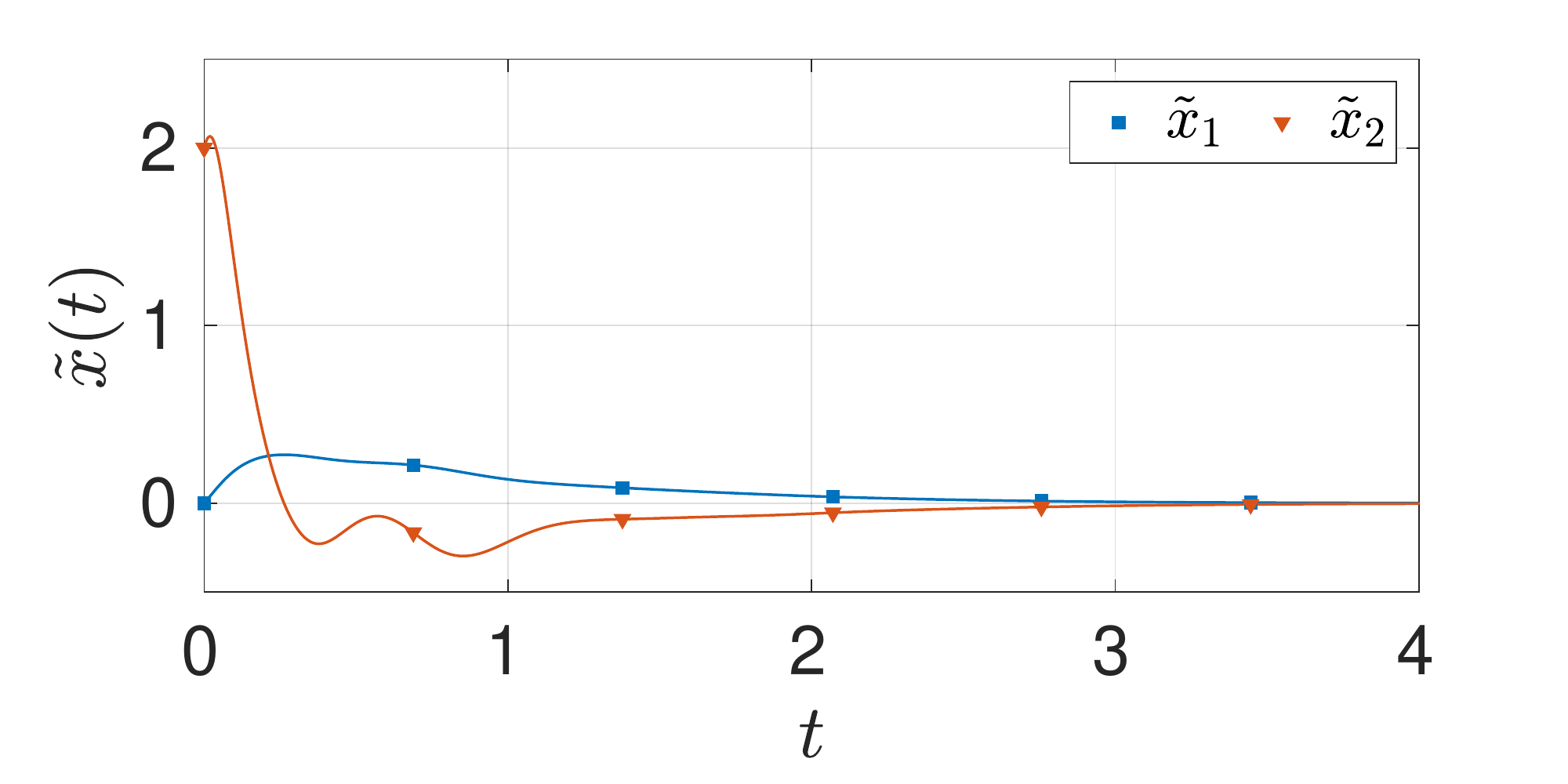}
		\caption{ Estimation errors between the original state variables and the estimated state variables under nominal gains from \eqref{table:SE_sim1} for the two-state dynamical system.}
		\label{fig: Estimation error_SE_sim1}
\end{figure}

\begin{figure}
		\includegraphics[width=0.75\columnwidth]{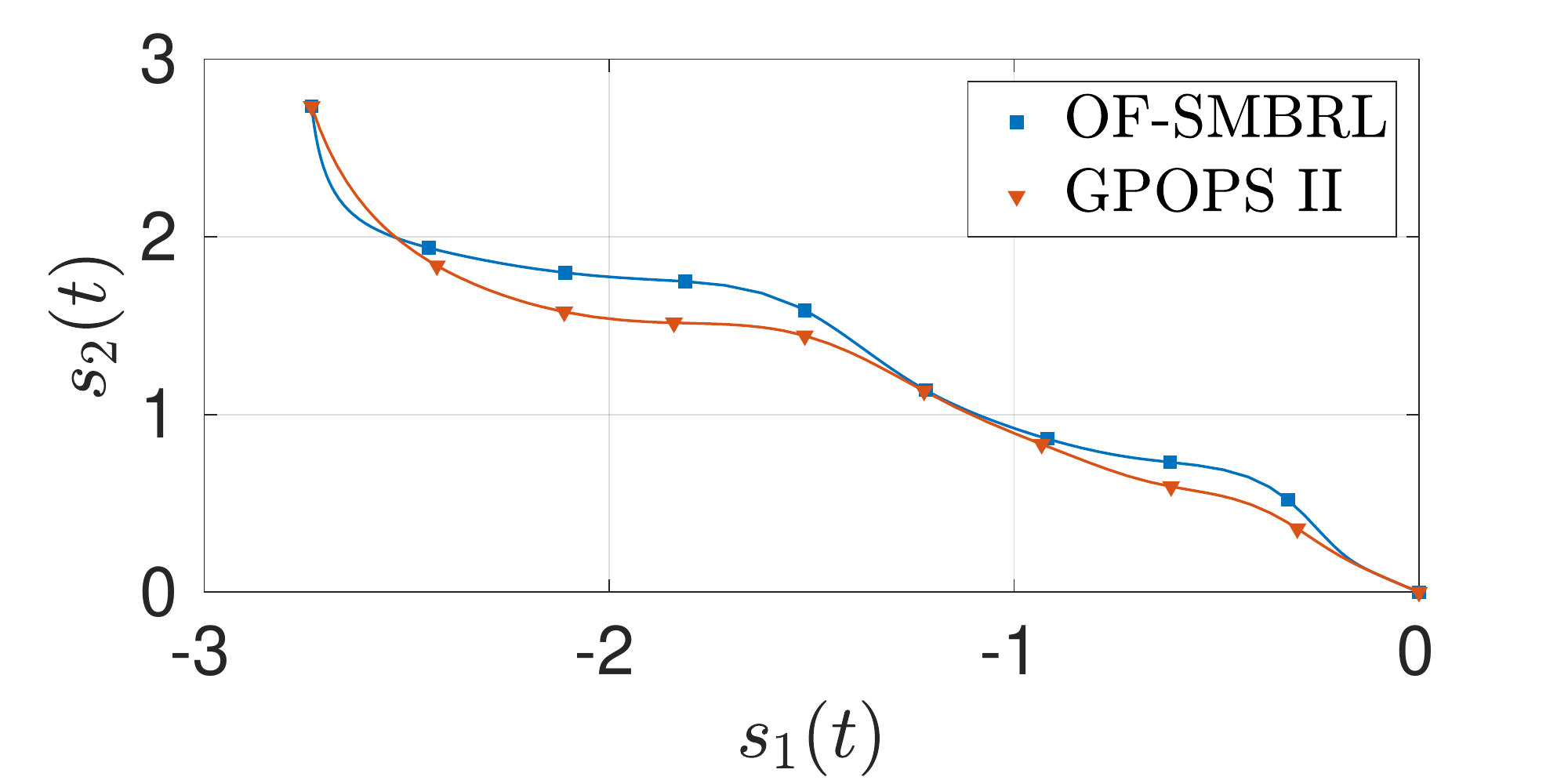}
		\caption{ Comparison of the optimal trajectories obtained using GPOPS II and OF-SMBRL with learned optimal weights for the two-state dynamical system.}
		\label{fig: Comparison_Optimal_SE}
\end{figure}

\subsubsection{Sensitivity Analysis for the two state system} \label{sen_sim1}
To study sensitivity of the developed technique to changes in various tuning gains, a one-at-a-time sensitivity analysis is performed. The gains $k$, $\alpha$, $\beta_{1}$, $k_{c}$, $k_{a1}$, $k_{a2}$, $\beta$, and $v$ are selected for the sensitivity analysis. The costs of the trajectories, under the optimal feedback controller obtained using the developed method, are presented in Table \ref{table:Sentivity_analysis_SE_sim1} for 5 different values of each gain. The gains are varied in a neighborhood of the nominal values (selected through trial and error\footnote{If the costs for different values of a particular gain are identical, then the value which yields the fastest computation is selected as the nominal value for that gain.}) $k = 10$, $\alpha = 1$, $\beta_{1} = 5$, $k_{c} = 5$, $k_{a1} = 100$, $k_{a2} = 0.1$, $\beta = 1$, and $v = 1$. 

The results in Table \ref{table:Sentivity_analysis_SE_sim1} indicate that the developed method is robust to small changes in the learning gains. 

\begin{table}
    \centering
    \caption{ Sensitivity Analysis for the two state system. The gains are varied in a neighborhood of the nominal values (selected through trial and error to balance performance and computation speed) $k = 10$, $\alpha = 1$, $\beta_{1} = 5$, $k_{c} = 5$, $k_{a1} = 100$, $k_{a2} = 0.1$, $\beta = 1$, $v = 1$ , and DS indicates diverged simulation.}
\label{table:Sentivity_analysis_SE_sim1}
\begin{tabular}{p{2cm}p{2cm}p{2cm}p{2cm}p{2cm}p{2cm}p{2cm}} 
  \hline 
 $k_{c}$=  & 1 & 1.5 &  5 & 50  & 60   \\
 \hline
 Cost   & DS & 55.822 &  55.820 & 56.19 & DS \\ 
  \hline
   \hline
   $k_{a_1}$=  & 40 & 50 & 100 & 1500 & 2000 \\
 \hline
 Cost & DS &  55.83 & 55.820 & 55.83 & DS \\
  \hline
   \hline
   $k_{a_2}$=  & 0.0001 & 0.001 & 0.1   & 10 & 15 \\
 \hline
 Cost   &55.820 &  55.820&  55.820 &  55.84 &  DS\\
  \hline 
  \hline
   $\beta$=  & 0.001 & 0.1 & 0.1  & 20 & 50 \\
 \hline
 Cost   &55.821 &  55.822&  55.820 &  55.83 &  DS\\
  \hline
   \hline
   $v$=  & 0.1 &  0.5 & 1  & 3 & 5\\
 \hline
 Cost & DS & 55.85 &  55.820 & 55.821  & DS\\
  \hline
  \hline
      $k$=  & 5 &  7  & 10   & 200  & 250 \\
 \hline
 Cost & DS & 55.821 & 55.820 & 55.821 & DS  \\
  \hline
  \hline
      $\alpha$= & 0.0001 &  0.001  & 1   & 10  & 15 \\
 \hline
Cost & DS  & 55.84 & 55.820 & 55.821 & DS \\
  \hline
  \hline
      $\beta_{1}$=  & 0.001 & 1 & 5  & 200 & 400 \\
 \hline
Cost & 55.820  & 55.820  & 55.820  & 55.822  & DS  \\
  \hline
   \hline
\end{tabular}
\end{table}

\subsection{Four state dynamical system}\label{simsec2} 

The four-state dynamical system corresponding to a two-link planar robot manipulator is given by
\begin{align}\label{eq:sim2_SE}
    \dot {x}_{1} = x_{2}, \quad
    \dot {x}_{2} = f(x)+g(x)u,
\end{align}
where $x_{1} \in \mathbb{R}^2$ is the measured output, $x_{2} \in \mathbb{R}^2$, $u \in \mathbb{R}^2$ and 
{\medmuskip=0mu\begin{gather*}
f(x) = -M_{a}^{-1}\Bigg(V_{M}\begin{bmatrix}
 x_{2_{1}}\\x_{2_{2}}\end{bmatrix}+\begin{bmatrix}
 f_{d_{1}}x_{2_{1}}+f_{s_{1}}\tanh(x_{2_{1}})\\f_{d_{2}}x_{2_{2}}+f_{s_{2}}\tanh(x_{2_{2}})\end{bmatrix}\Bigg), \quad
 g(x) = \begin{bmatrix}\begin{bmatrix}
  0,0
 \end{bmatrix}^{T},\begin{bmatrix}
  0,0
 \end{bmatrix}^{T},(M_{a}^{-1})^{T}\end{bmatrix}^{T},\\
D \coloneqq \mathrm{diag}
\begin{bmatrix}
 x_{2_{1}}, x_{2_{2}}, \tanh(x_{2_{1}}), \tanh(x_{2_{2}})
\end{bmatrix},\\
M_{a} \coloneqq 
\begin{bmatrix}
 p_{1}+2p_{3}c_{2} & p_{2}+p_{3}c_{2} \\ p_{2}+p_{3}c_{2} & p_{2}
\end{bmatrix}, \quad
V_{M} \coloneqq 
\begin{bmatrix}
 -p_{3}s_{2}x_{2_{2}} &  -p_{3}s_{2}( x_{2_{1}}+ x_{2_{2}}) \\  p_{3}s_{2} x_{2_{1}} & 0
\end{bmatrix},  
\end{gather*}} with $s_{2} = \sin(x_{1_{2}})$, $c_{2} = \cos(x_{1_{2}})$, $p_{1} = 3.473$, $p_{2} = 0.196$, $p_{3} = 0.242$. The parameters are selected as $f_{d_{1}} = 5.3 , f_{d_{2}} = 1.1, f_{s_{1}} = 8.45, f_{s_{1}}= 2.35$.

The state $x$ = $[x_{1_{1}} \  \ x_{1_{2}} \  \ x_{2_{1}} \  \ x_{2_{2}}]^{T}$ corresponds to the angular positions and the angular velocities of the two links. The objective for the controller is to minimize the infinite horizon cost function in (\ref{cost function}), with $Q'(s) = s^{T}Qs$ where $Q = 10I_{4}$ and $R = I_{2}$ while satisfying the constraints, 
$x_{1_{1}} \in (-1,1)$; $x_{1_{2}} \in (-1,1)$; $x_{2_{1}} \in (-2,2)$; $x_{2_{2}} \in (-2,2)$.
The basis functions for value function approximation are selected as $\sigma(\hat{s}) = 
[\hat{s}_{1_{1}}\hat{s}_{2_{1}} ;\hat{s}_{1_{2}}\hat{s}_{2_{2}} ;\hat{s}_{2_{1}}\hat{s}_{1_{2}} ;\hat{s}_{2_{2}}\hat{s}_{1_{1}} ;\hat{s}_{1_{1}}\hat{s}_{1_{2}} ;\hat{s}_{2_{2}}\hat{s}_{2_{1}} ;\hat{s}_{1_{1}}^{2} ; \hat{s}_{1_{2}}^{2}; \\ \hat{s}_{2_{1}}^{2} ;\hat{s}_{2_{2}}^{2}]$.

The initial conditions for the state variables, estimated state variables, and the initial guesses for the weights are selected as  $x(0) =  [-0.5; -0.5; 1; 1]$, $\hat{x}(0) =  [-0.5; -0.5; 1.1; 1.1]$, $\Gamma(0)= 10I_{10}$, $\hat{W}_{a}(0) = \left[5;15;0;0;10;2;15;5;2;2\right]$, and $\hat{W}_{c}(0) = \left[15;15;0;0;12;2;15;8;2;2\right]$.
The ideal values of the actor and the critic weights are unknown. The simulation uses 625 fixed Bellman error extrapolation points in a 0.9$\times$0.9 square around the origin of the $s-$coordinate system.

\subsubsection{Results for the four state system} \label{Result_sim2}

As seen from Fig. \ref{fig:original_state_robot}, the system state stays within the user-specified safe set while converging to the origin. As demonstrated in Fig. \ref{fig: Estimation error_SE_sim2} the state estimates converge to the true values. \\
\begin{table}
    \centering
    \caption{ Costs for a single barrier transformed trajectory of \eqref{eq:sim2_SE}, obtained using the developed method, and using pseudospectral numerical optimal control software.}
\label{table:SE_sim2_robot}
\begin{tabular}{p{6cm}p{2cm}} 
  \hline 
  Method & Cost \\
  \hline
OF-SMBRL  &  15.27\\
 GPOPS II & 11.68 \\
  \hline
\end{tabular}
\end{table}
A comparison with offline numerical optimal control, similar to the procedure used for the two-state system, yields the results in Table \ref{table:SE_sim2_robot} indicate that the two solution techniques generate slightly different trajectories in the state space (see Fig. \ref{chap4:fig: Comparison_Optimal_robot}). Moreover, as expected, the total cost of the trajectories from OF-SMBRL method is higher than the total cost of the trajectories from GPOPS. We hypothesize that the difference in costs is due to the basis for value function approximation being unknown.
In summary, the newly developed method can achieve online synthesis of optimal feedback control through a SMBRL approach while estimating the unknown state variables in the system dynamics and ensuring safety guarantees in the original coordinates.

\begin{figure}
        \centering
		\includegraphics[width=0.75\columnwidth]{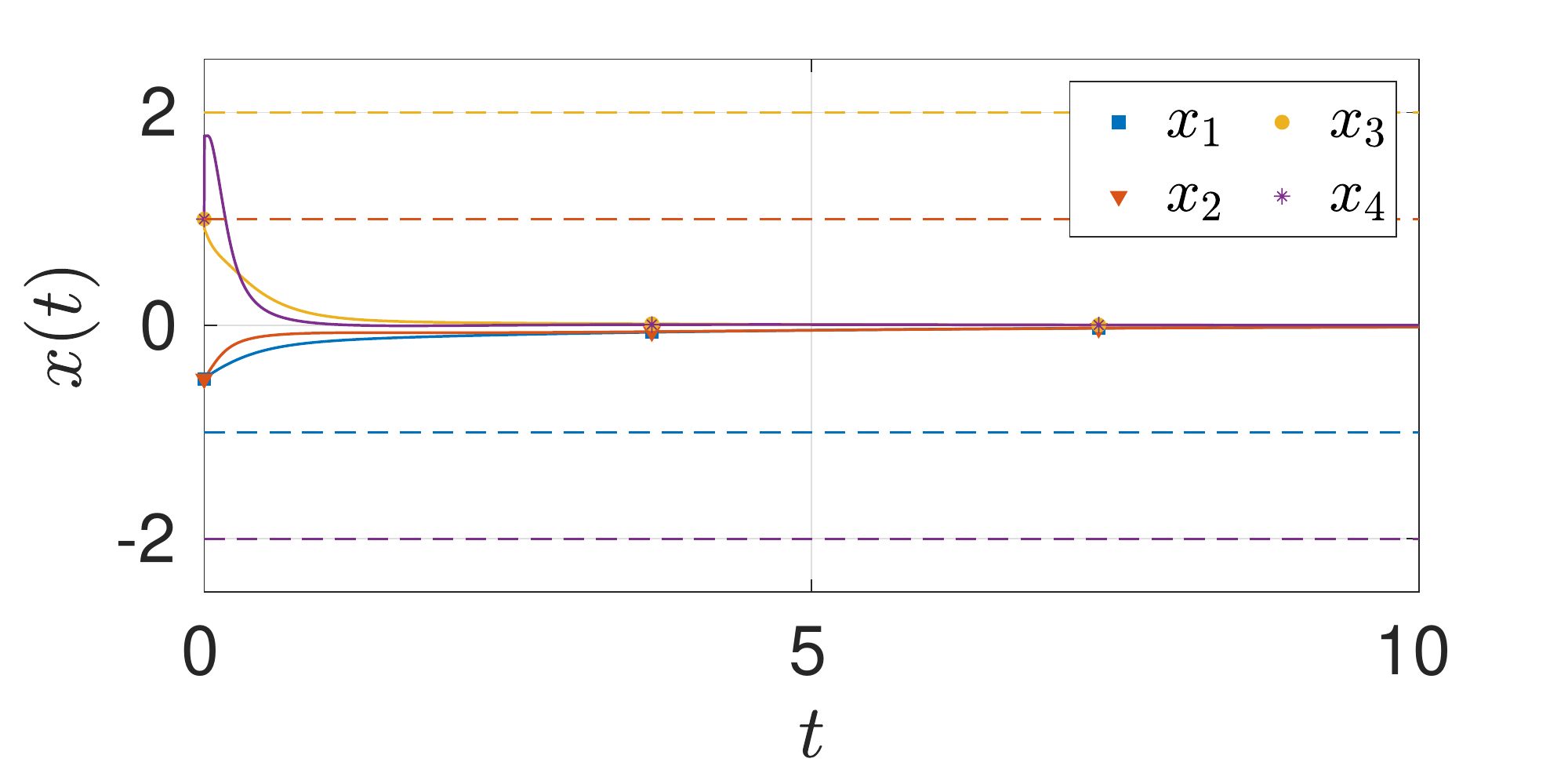}
		\caption {State trajectories for the four-state dynamical system using OF-SMBRL in the original coordinates. The dashed lines represent the user-selected safe set.} 
		\label{fig:original_state_robot}
\end{figure}

\begin{figure}
		\includegraphics[width=0.75\columnwidth]{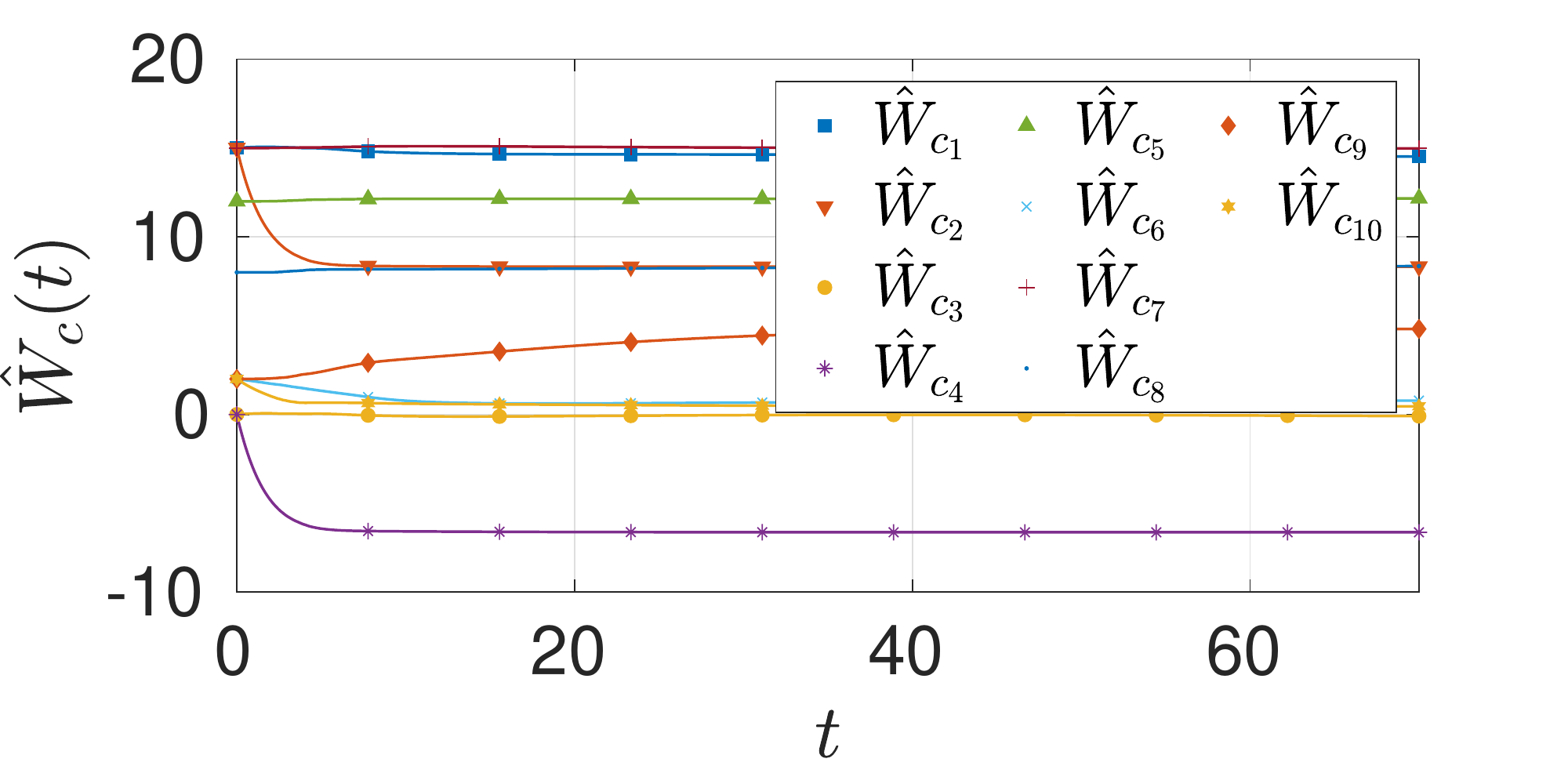}
		\caption{ Estimates of the critic weights under nominal gains for the four-state dynamical system.}
		\label{fig: Weight_SE_robot}
\end{figure}
\begin{figure}
		\includegraphics[width=0.75\columnwidth]{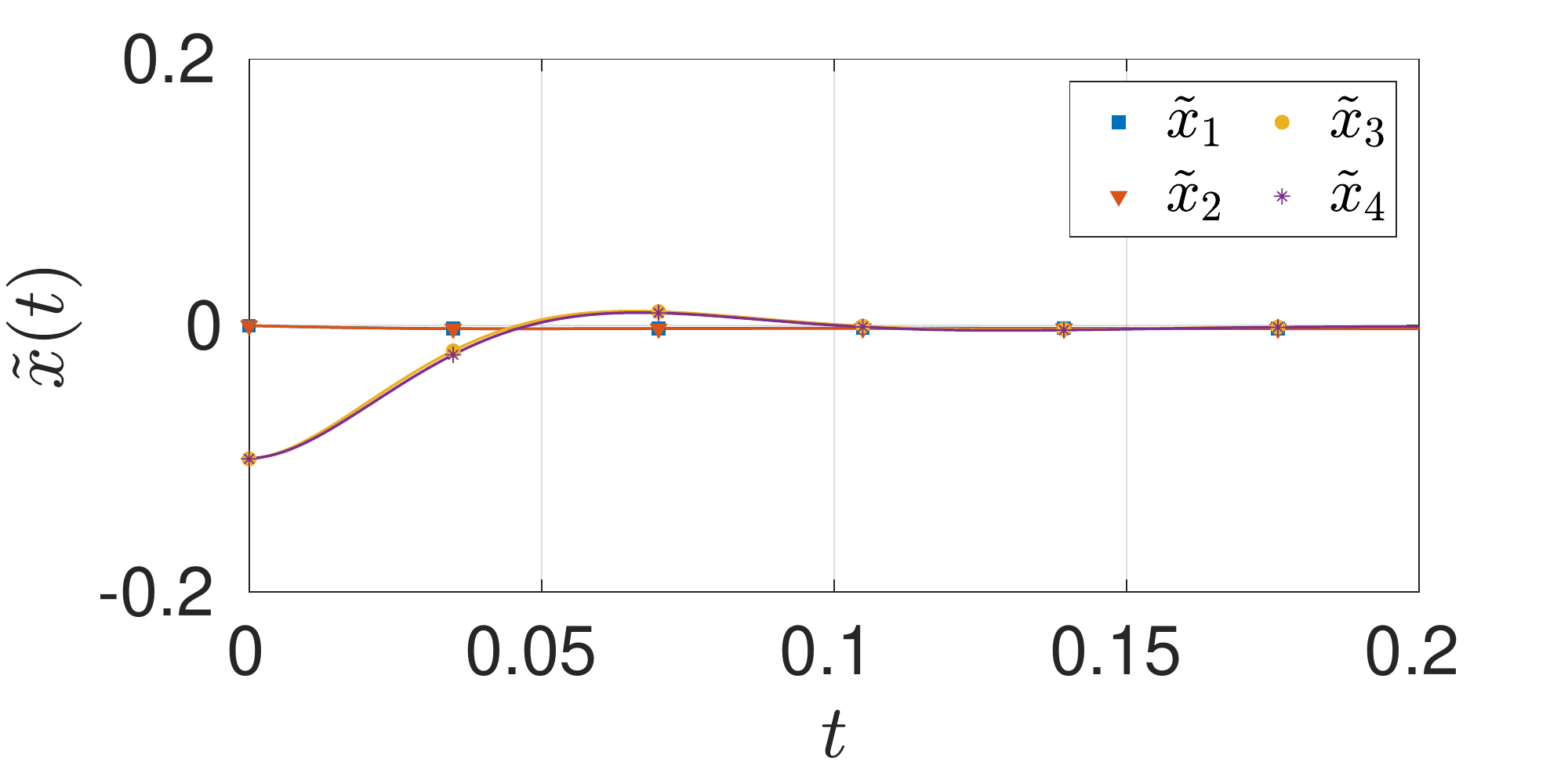}
		\caption{ State estimation errors under nominal gains for the four-state dynamical system.}
		\label{fig: Estimation error_SE_sim2}
\end{figure}


\begin{figure}
     \centering	
     \includegraphics[width=0.75\columnwidth]{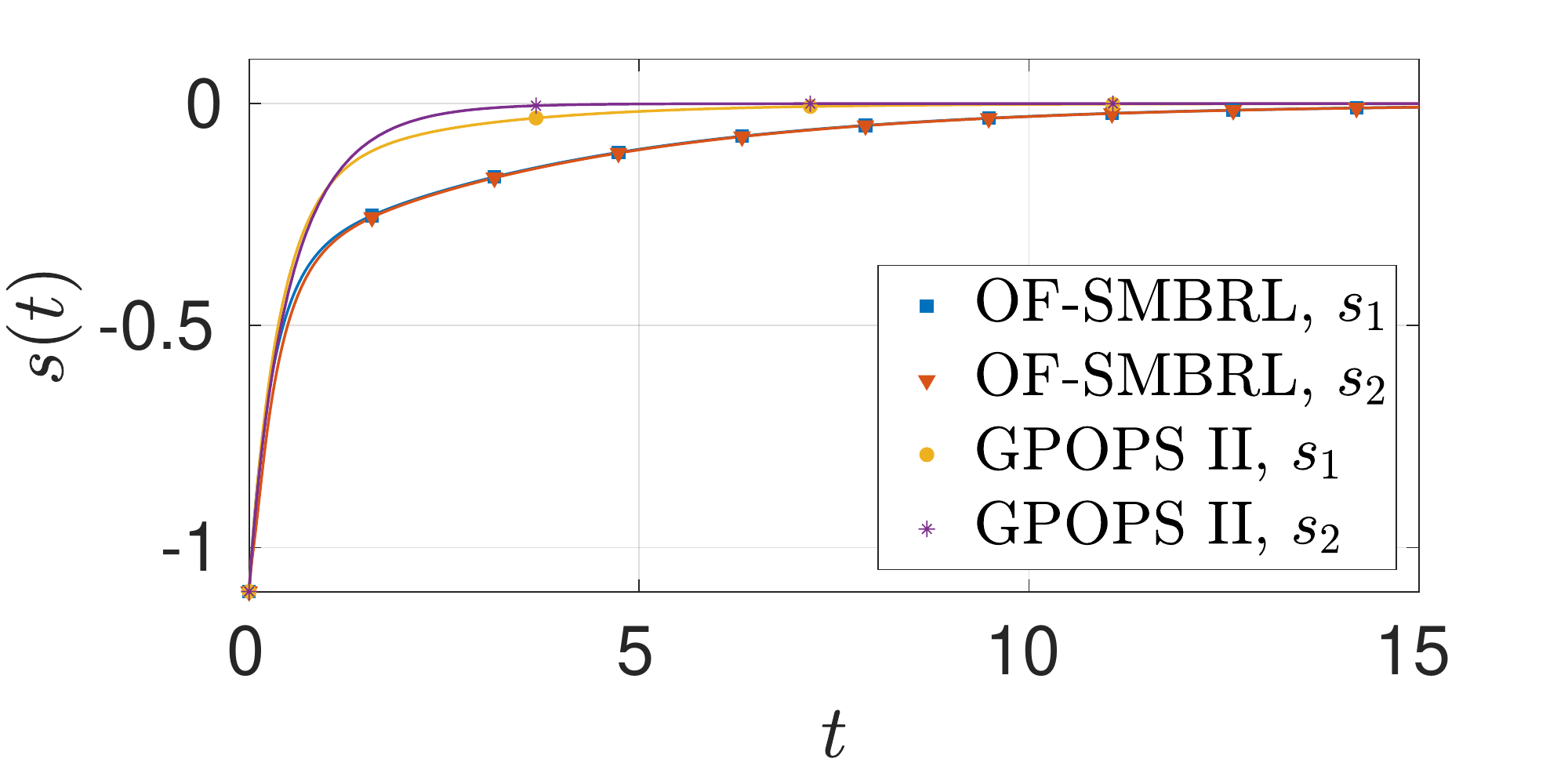}
     
     \includegraphics[width=0.75\columnwidth]{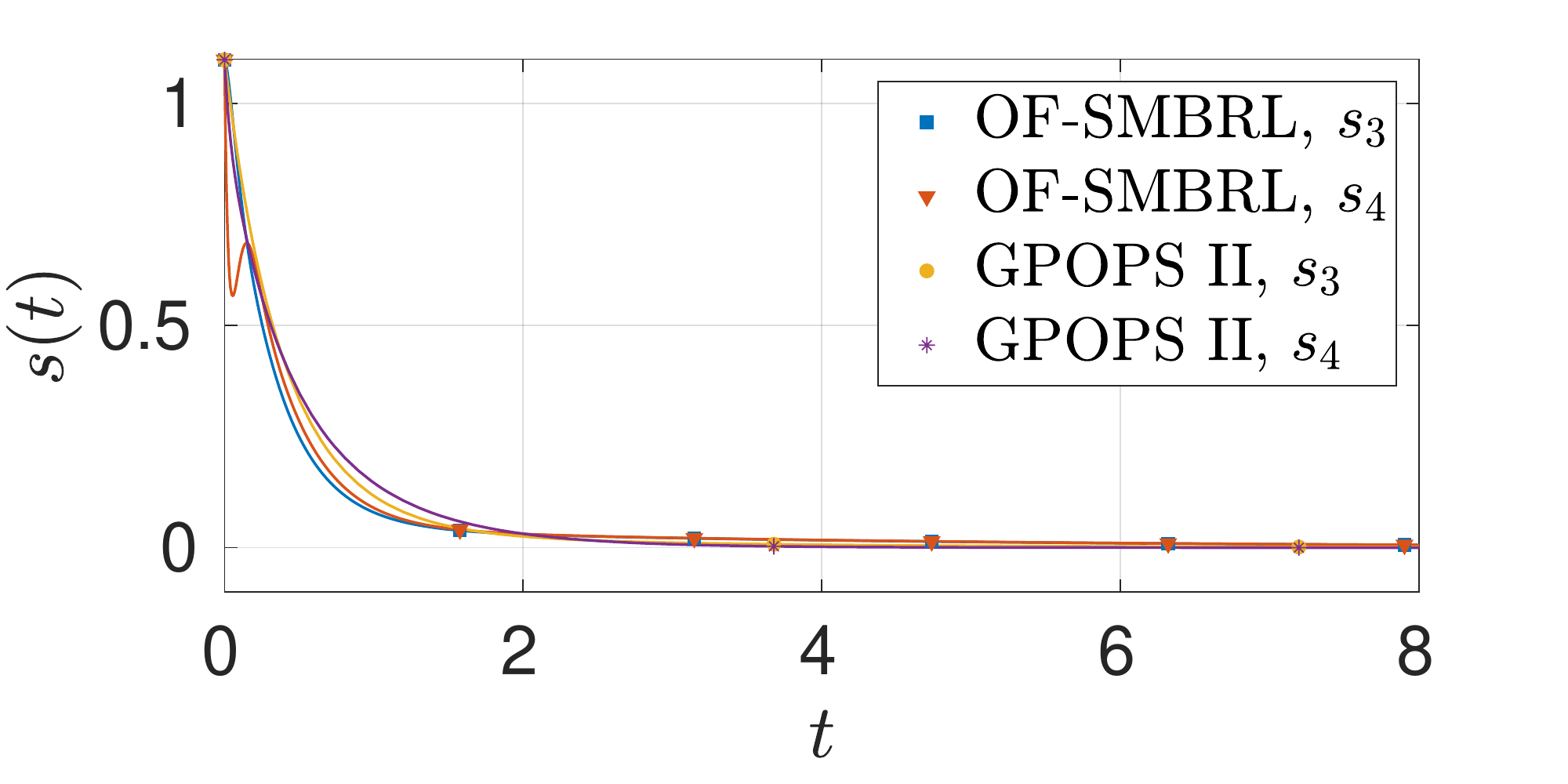}
     
	\caption{ Comparison of the optimal angular position (top) and angular velocity (bottom) trajectories obtained using GPOPS II and OF-SMBRL with fixed learned weights for the four-state dynamical system.}
    \label{chap4:fig: Comparison_Optimal_robot}
\end{figure}



The following section details a one-at-a-time sensitivity analysis and study the sensitivity of the developed technique to changes in various tuning gains.
\subsubsection{Sensitivity Analysis for the four state system}
The gains $k$, $\alpha$, $\beta_{1}$ ,$k_{c}$, $k_{a1}$, $k_{a2}$, $\beta$, and $v$ are selected for the sensitivity analysis. The costs of the trajectories, under the optimal feedback controller obtained using the developed method, are presented in Table \ref{table:SE_sim2} for 5 different values of each gain.
\begin{table}
    \centering
    \caption{ Sensitivity Analysis for the four state system. The gains are varied in a neighborhood of the nominal values (selected through trial and error to balance performance and computation speed) $k = 50$, $\alpha = 1$, $\beta_{1} = 10$, $k_{c} = 1000$, $k_{a1} = 100$, $k_{a2} = 0.5$, $\beta = 0.001$, $v = 500$; DS indicates diverged simulation.} 
\label{table:SE_sim2}
\begin{tabular}{p{2cm}p{2cm}p{2cm}p{2cm}p{2cm}p{2cm}p{2cm}} 
  \hline 
 $k_{c}$=  & 1 & 100 &  1000 & 2000  & 5000   \\
 \hline
 Cost   & \text{DS} & 16.70 &  15.27 & 15.28 & \text{DS} \\ 
  \hline
  \hline
  $k_{a_1}$=  & 10 & 50 & 100 & 250 & 500  \\
 \hline
 Cost & \text{DS} &  \text{DS} & 15.27 & 15.31 & 15.32 \\
  \hline
  \hline
  $k_{a_2}$=  & 0.01 & 0.1 & 0.5 & 1 & 2  \\
 \hline
 Cost   &15.28&  15.28& 15.27&  \text{DS} &  \text{DS} \\
  \hline 
  \hline
  $\beta$=  & 0.00001 & 0.0001 &  0.001 & 0.01   & 0.1   \\
 \hline
 Cost   &\text{DS}  & \text{DS} & 15.27 & 15.32&\text{DS}\\
  \hline
  \hline
  $v$=  & 200 &  300 & 500  & 600 & 1000 \\
 \hline
 Cost &\text{DS}&15.29 &  15.27 & \text{DS} & \text{DS}  \\
  \hline
  \hline
      $k$=  & 0.01 &  1 & 50  & 100 & 500 \\
 \hline
 Cost & \text{DS} &  15.27 &  15.27 &  15.27 & \text{DS}  \\
  \hline
  \hline
      $\alpha$= & 0.01 &  0.1 & 1  & 20 & 100 \\
 \hline
Cost &\text{DS} & \text{DS} & 15.27 & 15.27 & \text{DS} \\
  \hline
  \hline
      $\beta_{1}$= & 1 &  5 & 10  & 20 & 100 \\
 \hline
Cost & \text{DS} & 15.27 & 15.27 & 15.27 & \text{DS}  \\
  \hline
  \hline
\end{tabular}
\end{table}

The gains are varied in a neighborhood of the nominal values (selected through trial and error) $k_{c} = 1000$, $k_{a1} = 100$, $k_{a2} = 0.5$, $\beta = 0.001$, $v = 500$,  $k = 50$, $\alpha = 1$, and $\beta_{1} = 10$. The results in Table \ref{table:SE_sim2}  indicate that the developed method is not sensitive to small changes in the learning gains.

The results in Table \ref{table:Sentivity_analysis_SE_sim1} and Table \ref{table:SE_sim2} show that the developed method is not sensitive to small changes in the learning gains, which simplifies gain selection. On the other hand, the developed method is sensitive to selection of basis functions and the initial guess of the unknown weights as indicated by the local stability result. Since the dimensionality of the vector of unknown weights is high, a complete characterization of the region of attraction is computationally difficult. Therefore, the basis functions and the initial conditions are selected via trial and error in this paper.

To our knowledge, the technique developed in this paper is the first one to address safe online optimal control of nonlinear continuous-time systems under output feedback. As such, meaningful simulations to compare with existing methods were not possible. However, the performance of the developed method is verified by comparing the total cost of a trajectory starting from a given initial condition under the learned controller with the optimal cost starting from the same initial condition, computed using pseudospectral optimization (GPOPS-II \cite{SCC.GPOPS}) with full state feedback. Since the developed SMBRL technique relies on output feedback, the total cost for OF-SMBRL is expected to be higher than the optimal cost under state feedback, as confirmed by the simulation results in Table \ref{table:SE_sim1} and Table \ref{table:SE_sim2_robot}.






\section{Conclusion}

This paper presents a novel framework for online approximate optimal control of a class of safety-critical nonlinear systems. The framework consists of a novel state estimator and a safe MBRL based controller. A BT is used to transform the constrained optimal control problem into an unconstrained optimal control problem. MBRL is used to solve the problem online in the transformed coordinates in conjunction with the novel state estimator to estimate the transformed state variables. Regulation of the system state variables to a neighborhood of the origin and convergence of the estimated policy to a neighborhood of the optimal policy is established using a Lyapunov-based stability analysis. The OF-SMBRL controller is guaranteed to keep the state of the original system within the safety bounds.

The barrier function used in the BT to address safety can only constrain the state variables of the system to a box. A more generic and adaptive barrier function, constructed, perhaps, using sensor data, is a subject for future research. The barrier transformation method to ensure safety relies on knowledge of the dynamics of the system. In particular, the safety guarantees rely on the relationship  between trajectories of the original dynamics and the transformed system (Lemma \ref{lem:trajectoryRelation}) and the relationship  between the trajectories of the state estimator in the transformed and the original coordinates (Lemma \ref{lem2}). Both relationships fail to hold if a part of the dynamics is not included in the original model or if there are uncertain parameters in the model. While parametric uncertainties can be addressed under full state feedback (see, e.g., \cite{SCC.Mahmud.Nivison.eatoappear}), further research is needed to establish safety guarantees for the output feedback case with parametric uncertainties and/or unmodeled dynamics (for a differential games approach to robust safety, see \cite{SCC.Yang.Ding.ea2020}).

\bibliographystyle{IEEETrans.bst}
\bibliography{scc,sccmaster,scctemp}

\appendix

\section{Appendix}

\subsection{Proof of Lemma \ref{lem:trajectoryRelation}}\label{appendix:lemma1}
\begin{manuallemma}{\ref{lem:trajectoryRelation}}
If $t \mapsto \Phi\big(t,b(x^{0}),\zeta\big)$ is a complete Carath\'{e}odory solution to \eqref{eq:BTDynamics}, starting from the initial condition $b(x^{0})$, under the feedback policy $(s,t) \mapsto \zeta (s,t)$ and $t \mapsto \Lambda(t,x^{0},\xi)$ is a Carath\'{e}odory solution to \eqref{eq:Dynamics}, starting from the initial condition $x^{0}$, under the feedback policy $(x,t) \mapsto \xi(x,t)$, defined as $\xi(x,t) = \zeta(b(x),t)$, then $t \mapsto \Lambda(t,x^{0}, \xi)$ is complete and $ \Lambda(t,x^{0}, \xi) = b^{-1}\left(\Phi(t,b(x^{0}),\zeta)\right) $ for all $t \in \mathbb{R}_{\geq 0}$.
\end{manuallemma}
\begin{proof}
	Note that since $t \mapsto \Phi\big(t,b(x^{0}),\zeta\big)$ is a complete Carath\'{e}odory solution to $ \dot{s} = F(s) + G(s) \zeta(s,t) $, it is differentiable at almost all $ t \in \mathbb{R}_{\geq 0} $. Since $ b^{-1} $ is smooth, $ t \mapsto b^{-1}\left(\Phi\left(t,b(x^{0}),\zeta\right)\right) $ is also differentiable at almost all $ t \in \mathbb{R}_{\geq 0} $. That is,
	\begin{equation}
		\frac{\mathrm{d}(b^{-1}\circ\Phi_i)}{\mathrm{d}t}\left(t,b(x^{0}),\zeta\right) = \frac{\mathrm{d} b^{-1}_{(a_i,A_i)}(y)}{\mathrm{d} y}|_{y=\Phi_{i}\left(t,b\left(x^{0}\right), \zeta\right)} \frac{\mathrm{d} \Phi_{i}}{\mathrm{d} t}\left(t,b(x^{0}),\zeta\right), 
\end{equation}
	for almost all $t \in \mathbb{R}_{\geq 0}$ and all $i=1,\cdots,n$, where $\Phi_i$ denotes the $i$th component of $\Phi$. As a result,
	\begin{align*}
		\frac{\mathrm{d}(b^{-1}\circ\Phi_i)}{\mathrm{d}t}\left(t,b(x^{0}),\zeta\right) =  \frac{\left(F\left(\Phi\left(t,b(x^{0}),\zeta\right)\right)\right)_i}{B_i\left(\Phi_i\left(t,b(x^{0}),\zeta\right)\right)} 
		+ \frac{\left(G\left(\Phi\left(t,b(x^{0}),\zeta \right)\right)\right)_i \zeta\left(\Phi\left(t,b(x^{0}),\zeta\right),t\right)}{B_i\left(\Phi_i\left(t,b(x^{0}),\zeta\right)\right)},
\end{align*}
		for almost all $ t \in \mathbb{R}_{\geq 0}$, and all $i=1,\hdots,n$. By the construction of $ F $, $ G $, and $ \xi $,
	\begin{multline*}
		\frac{\mathrm{d}(b^{-1}\circ\Phi)}{\mathrm{d}t}\left(t,b(x^{0}),\zeta\right)  = f\left(b^{-1}\circ\Phi\left(t,b(x^{0}),\zeta\right)\right)\theta\\+g\left(b^{-1}\circ\Phi\left(t,b(x^{0}),\zeta\right)\right) \xi\left(b^{-1}\circ\Phi\left(t,b(x^{0}),\zeta\right),t\right),
	\end{multline*}
	for almost all $ t \in \mathbb{R}_{\geq 0} $. Clearly $t \mapsto b^{-1}\circ\Phi\left(t,b(x^{0}),\zeta\right)$ is a Carath\'{e}odory solution of \eqref{eq:Dynamics} on $\mathbb{R}_{\geq 0}$, starting from the initial condition $ b^{-1}\big(b(x^{0})\big)= x^{0}$ under the feedback policy $ (x,t)\mapsto \xi(x,t) $. By uniqueness of solutions $ \dot{x} = f(x)\theta + g(x) \xi(x,t) $ (which follows from local Lipschitz continuity of $ f $, $ g $, and $ b $ inside the barrier), $\Lambda(\cdot,x^0,\xi)$ is complete and $\Lambda(t,x^0,\xi) = b^{-1}\left(\Phi\left(t,b(x^{0}),\zeta\right)\right)$ for all $ t \in \mathbb{R}_{\geq 0} $.
\end{proof}
\subsection{Proof of Lemma \ref{lem2}}\label{appendix:lemma2}
\begin{manuallemma}{\ref{lem2}}
If $t \mapsto \Psi\big(t;b(x_{1}(\cdot)),b(\hat{x}^{0}) \big)$ is a Carath\'{e}odory solution to \eqref{eq:BTDynamicsestimated} along the trajectory $x_{1}(\cdot)$, starting from the initial condition $b(\hat{x}^{0})$, and if $t \mapsto \xi(t;x_{1}(\cdot),\hat{x}^{0})$ is a Carath\'{e}odory solution to \eqref{eq:baseEqx}, starting from the initial condition $\hat{x}^{0}$, along the trajectory $x_{1}(\cdot)$, then $ \xi(t;x_{1}(\cdot),\hat{x}^{0}) = b^{-1}\big(\Psi\big(t;b(x_{1}(\cdot)),b(\hat{x}^{0})\big)\big) $ for all $t \in \mathbb{R}_{\geq 0}$.
\end{manuallemma}
\begin{proof}
Since $t \mapsto \Psi\big(t;b(x_{1}(\cdot)),b(\hat{x}^{0}) \big)$ is a Carath\'{e}odory solution to \eqref{eq:BTDynamicsestimated},
it is differentiable at almost all $t$. Since $b^{-1}$ is smooth, $t \mapsto b^{-1}\big(\Psi\big(t;b(x_{1}(\cdot)),b(\hat{x}^{0})\big)\big)$ is also differentiable at almost all $t$. When $b^{-1}\big(\Psi\big(t;b(x_{1}(\cdot)),b(\hat{x}^{0})\big)\big)$ is differentiable,
\begin{equation*}
  \frac{\mathrm{d}}{\mathrm{d}t}\big(b^{-1}\big(\Psi\big(t;b(x_{1}(\cdot)),b(\hat{x}^{0})\big)\big)\big) =\frac{\mathrm{d} \big(b^{-1}\big(\Psi\big(t;b(x_{1}(\cdot)),b(\hat{x}^{0})\big)\big)\big)}{\mathrm{d} s}\frac{\mathrm{d} s}{\mathrm{d} t}. 
\end{equation*}
So, 
\begin{multline*}
  \frac{\mathrm{d}}{\mathrm{d}t}\big(b^{-1}\big(\Psi\big(t;b(x_{1}(\cdot)),b(\hat{x}^{0})\big)\big)\big)=\frac{\mathrm{d}\big(b^{-1}\big(\Psi\big(t;b(x_{1}(\cdot)),b(\hat{x}^{0})\big)\big)\big)}{\mathrm{d} s} \\  \begin{bmatrix}
     H\big(\Psi\big(t;b(x_{1}(\cdot)),b(\hat{x}^{0})\big)\big) \\
    F\big(\Psi\big(t;b(x_{1}(\cdot)),b(\hat{x}^{0})\big)\big) + G\big(\Psi\big(t;b(x_{1}(\cdot)),b(\hat{x}^{0})\big)\big)u(t)  \\ + \nu_{2}\big(\Psi\big(t;b(x_{1}(\cdot)),b(\hat{x}^{0})\big),t\big)
\end{bmatrix}. 
\end{multline*}

By the construction of $H$, $F$, $\nu_{2}$ and $G$, for almost all $t \in \mathbb{R}_{\geq 0}$,
\begin{equation*}
  \frac{\mathrm{d}}{\mathrm{d}t}\big(b^{-1}\big(\Psi\big(t;b(x_{1}(\cdot)),b(\hat{x}^{0})\big)\big)\big) = \begin{bmatrix}
     b^{-1}\big(\Psi_{2}\big(t;b(x_{1}(\cdot)),b(\hat{x}_{2}^{0})\big)\big) \\
    f\big(b^{-1}\big(\Psi\big(t;b(x_{1}(\cdot)),b(\hat{x}^{0})\big)\big)\big) \\ + g\big(b^{-1}\big(\Psi\big(t;b(x_{1}(\cdot)),b(\hat{x}^{0})\big)\big)\big)u(t)  \\ +
    \nu_{1}\big(b^{-1}\big(\Psi\big(t;b(x_{1}(\cdot)),b(\hat{x}^{0})\big),t\big)\big)
\end{bmatrix}.
\end{equation*}
Clearly $t \mapsto b^{-1}\big(\Psi\big(t;b(x_{1}(\cdot)),b(\hat{x}^{0})\big)\big)$ is a Carath\'{e}odory solution to \eqref{eq:baseEqx}, starting from the initial condition $ b^{-1}\big(b(\hat{x}^{0})\big)= \hat{x}^{0}$ along the trajectory $t \mapsto b(x_{1}(t))$. Finally, continuity of $t \mapsto b^{-1}\big(\Psi\big(t;b(x_{1}(\cdot)),b(\hat{x}^{0})\big)\big)$ and $t \mapsto \xi(t;x_{1}(\cdot),\hat{x}^{0})$ implies that  $ b^{-1}\big(\Psi\big(t;b(x_{1}(\cdot)),b(\hat{x}^{0})\big)\big)= \xi(t;x_{1}(\cdot),\hat{x}^{0}) $ for all $t \in \mathbb{R}_{\geq 0}$.
\end{proof}
\subsection{Proof of Lemma \ref{lem3}}\label{appendix:lemma3}
\begin{manuallemma}{\ref{lem3}}
Let $V_{se} : \mathbb{R}^{3n} \rightarrow \mathbb{R}_{\geq 0}$ be a continuously differentiable candidate Lyapunov function defined as
$ V_{se}(Z_{1}) \coloneqq \frac{\alpha^2}{2}\tilde{s}_{1}^T\tilde{s}_{1} + \frac{1}{2}r^Tr + \frac{1}{2}\eta^T\eta$,
where $Z_{1} \coloneqq [\tilde{s}_{1}^T,r^T,\eta^T]$. Provided $s$, $\hat{s} \in \overline{B}(0,\chi)$ for some $\chi > 0$, the orbital derivative of $V_{se}$ along the trajectories of $\dot{\tilde{s}}_{1}$, $\dot{r}$, and $\dot{\eta}$, defined as $\dot{V}_{se}(Z_{1},s,\tilde{s},\tilde{W}_{a}) \coloneqq \frac{\partial V_{se}(Z_{1},s,\tilde{s},\tilde{W}_{a})}{\partial \tilde{s}_{1}}(H(s)-H(\hat{s})) + \frac{\partial V_{se}(Z_{1},s,\tilde{s},\tilde{W}_{a})}{\partial r}\dot{r} + \frac{\partial V_{se}(Z_{1},s,\tilde{s},\tilde{W}_{a})}{\partial \eta}\dot{\eta}$, can be bounded as $\dot{V}_{se}(Z_{1},s,\tilde{s},\tilde{W}_{a}) \leq -{\alpha^3}\|\tilde{s}_{1}\|^{2} - (k-\varpi_{1}\varpi_{4})\|r\|^{2} - (\beta_{1}-\alpha)\|\eta\|^{2} + \\ \varpi_{1}\left(1+\varpi_{4}+\varpi_{4}\alpha\|\right) \|r\|\|\tilde{s}_{1}\|  + \varpi_{1}\varpi_{4} \|r\| \|\eta\|  + \varpi_{2} \|r\| \|\tilde{W}_{a}\|  +  \varpi_{3} \|r\|$.
\end{manuallemma}
\begin{proof}
Using \eqref{eq:BTDynamics}, first state of \eqref{eq:BTDynamicsestimated}, \eqref{eq:eta_final}, and \eqref{eq:r_totalfinal}, the orbital derivative can be expressed as
\begin{align}
\dot{V}_{se} = {\alpha^2}\tilde{s_{1}}^T(r - \alpha \tilde{s_{1}} - \eta) + r^{T}\dot{r} + \eta^{T}(-\beta_{1}\eta-kr-\alpha\dot{\tilde{s}}_1), \label{eq:derivative_Lyapunov_function1}
\end{align}
So, from (\ref{eq:derivative_Lyapunov_function1}),
\begin{multline}
\dot{V}_{se}(Z_{1},s,\tilde{s},\tilde{W}_{a}) = -{\alpha^3}\tilde{s_{1}}^T\tilde{s_{1}} - kr^Tr - (\beta_{1}-\alpha)\eta^T\eta + (r^T\tilde{F}_{2}(s,\hat{s}) + r^T\tilde{F}_{3}(s,\hat{s}) +  r^T \tilde{G}_{1}(s,\hat{s})\hat{u}).
\end{multline}
Using the Cauchy-Schwarz inequality and the fact that $F_{2}$, $F_{3}$, and $G$ are Lipschitz continuous on $\overline{B}(0,\chi)$, we get
\begin{multline}
\dot{V}_{se}(Z_{1},s,\tilde{s},\tilde{W}_{a}) \leq -{\alpha^3}\tilde{s_{1}}^T\tilde{s_{1}} - kr^Tr - (\beta_{1}-\alpha)\eta^T\eta + \varpi_{1} \|r\| \|\tilde{s}\| + \varpi_{2} \|r\| \|\tilde{W}_{a}\| +  \varpi_{3} \|r\|. \label{eq:derivative_Lyapunov_function2}
\end{multline}
Provided $s,\hat{s} \in \overline{B}(0,\chi)$ from \eqref{eq:BTDynamics} and \eqref{eq:BTDynamicsestimated}, 
\begin{align*}
s_2 = b\left(\dot{s}_{1} (\frac{A_{1}a_{1}^{2} - a_{1}A_{1}^{2}}{a_1^{2}e^{s_1}-2a_1A_1+A_{1}^{2}e^{-s_1}})\right) = h(s_1,\dot s_1),
\end{align*}
\begin{align*}
\hat{s}_{2}= b\left(\dot{\hat{s}}_{1} (\frac{A_{1}a_{1}^{2} - a_{1}A_{1}^{2}}{a_1^{2}e^{\hat{s}_1}-2a_1A_1+A_{1}^{2}e^{-\hat{s}_1}})\right) = h(\hat{s}_1,\dot{\hat s}_1),
  s_2 = b\left(\dot{s}_{1} (\frac{1}{B_{1}(s_{1})})\right) = h(s_1,\dot s_1),
\end{align*}
\begin{equation}
\hat{s}_{2}= b\left(\dot{\hat{s}}_{1} (\frac{1}{B_{1}(\hat{s}_{1})})\right) = h(\hat{s}_1,\dot{\hat s}_1),
\quad \text{and} \quad 
    \tilde{s}_{2} = s_{2} - \hat{s}_{2} = h(s_{1},\dot{s}_{1}) - h(\hat{s}_{1},\dot {\hat{s}}_{1}).\label{eq:tilde(s_2)}
\end{equation}
Provided $s,\hat{s} \in \overline{B}(0,\chi)$, Lipschitz continuity of $h$ can be exploited to derive the bound
\begin{multline}
       | h(s_{1},\dot{s}_{1}) - h(\hat{s}_{1},\dot {\hat{s}}_{1})|  \leq | h(s_{1},\dot{s}_{1}) - h(\hat{s}_{1},\dot{s}_{1})| \\ + | h(\hat{s}_{1},\dot{s}_{1}) - h(\hat{s}_{1},\dot {\hat{s}}_{1})| \leq \varpi_{4} \|\tilde{s}_{1}\| + \varpi_{4} \| \dot{\tilde{s}}_{1}\| \\ \leq \varpi_{4} \|\tilde{s}_1 \| + \varpi_{4} \| r-\alpha\tilde{s}_{1}-\eta \|.
\end{multline}
Using the triangle inequality, 
\begin{equation}\label{eq:need_nu1_design}
   \|\tilde{s}\| \leq \|\tilde{s_{1}}\| + \|\tilde{s_{2}}\| \leq(1 + \varpi_{4} + \varpi_{4}\alpha)\|\tilde{s}_{1}\| + \varpi_{4}\|r\| + \varpi_{4}\|\eta\|. 
\end{equation}
substituting \eqref{eq:need_nu1_design} in \eqref{eq:derivative_Lyapunov_function2}, we obtain the desired bound
\begin{multline}
\dot{V}_{se} \leq -{\alpha^3}\|\tilde{s}_{1}\|^{2} - (k-\varpi_{1}\varpi_{4})\|r\|^{2} - (\beta_{1}-\alpha)\|\eta\|^{2} \\  + \varpi_{1}\left(1+\varpi_{4}+\varpi_{4}\alpha\|\right) \|r\|\|\tilde{s}_{1}\|    + \varpi_{1}\varpi_{4} \|r\| \|\eta\| \\ + \varpi_{2} \|r\| \|\tilde{W}_{a}\| +  \varpi_{3} \|r\|.\label{eq:derivative_Lyapunov_function2e}
\end{multline}
\end{proof}
\subsection{Proof of Theorem \ref{thm:Optimal GAS}}\label{appendix:thm:Optimal GAS}
\begin{manualtheorem}{\ref{thm:Optimal GAS}}
	If the optimal state feedback controller \eqref{eq:optimalcontrol} that minimizes the cost function in (cost function) exists and if the corresponding optimal value function is continuously differentiable and radially unbounded, then the origin of closed-loop system 
	\begin{align}
	\dot{s}_{1}  = H(s),\quad
	\dot{s}_{2}  =F({s}) + G({s})u^{*}(s)\label{eq:Closed-loop System}
	\end{align} is globally asymptotically stable.
\end{manualtheorem}
\begin{proof}
Under the hypothesis of Theorem \ref{thm:Optimal GAS}, the optimal value function satisfies the Hamilton-Jacobi-Bellman equation \cite[Chapter 5]{SCC.Liberzon2012} 
\begin{equation}
    V^{*}_{s_{1}}\left(s\right)H(s_{1},s_{2})+V^{*}_{s_{2}}\left(s\right)\left(F\left(s\right)+G\left(s\right)u^{*}\left(s\right)\right)+c\left(s,u^{*}\left(s\right)\right)=0,\label{eq:HJB}
\end{equation}
with\begin{equation}
    u^{*}(s) := -\frac{1}{2}R^{-1}G(s)^{T}(\nabla_{s_{2}}V^{*}(s))^{T},
\end{equation}
Since the solutions of \eqref{eq:Closed-loop System} are continuous and the function $ V^{*} $ is positive semidefinite by definition, if $ V^{*}\left(\begin{bmatrix}
s_{1}\\s_{2}
\end{bmatrix}\right) = 0$ for some $ s \neq 0 $, it can be concluded that $ Q\left(\phi\left(t,s,u^{*}\left(\cdot\right)\right)\right) = 0, \forall t\geq 0$, and $ u^{*}\left(\phi\left(t,s,u^{*}\left(\cdot\right)\right)\right)= 0, \forall t\geq 0 $. If Assumption \ref{ass:CostRestrictions}-(a) holds then $ \phi\left(t,s,u^{*}\left(\cdot\right)\right)=0,\forall t\geq 0$, which contradicts $ s\neq 0 $. If Assumption \ref{ass:CostRestrictions}-(b) holds, then $ s_{1}\left(t,s,u^{*}\left(\cdot\right)\right)=0,\forall t\geq 0 $. As a result, $ \phi\left(t,s,u^{*}\left(\cdot\right)\right)=0,\forall t\geq 0$, which contradicts $ s\neq 0 $. If Assumption \ref{ass:CostRestrictions}-(c) holds, then $ s_{2}\left(t,s,u^{*}\left(\cdot\right)\right)=0,\forall t\geq 0 $. As a result, $ s_{1}\left(t,s,u^{*}\left(\cdot\right)\right)=c_{2},\forall t\geq 0 $ for some constant $ c_{2}\in\mathbb{R}^{n} $. Since $ F\left(s\right)\neq 0 $ if $ s_{1}\neq 0 $, it can be concluded that $ c_{2}=0 $, which contradicts $ s\neq 0 $. Hence, $ V^{*}\left(s\right) $ cannot be zero for a nonzero $ s $. Furthermore, since $ F\left(0\right)=0 $, the zero controller is clearly the optimal controller starting from $ s=0 $. That is, $ V^{*}\left(0\right)=0 $, and as a result, $ V^{*}:\mathbb{R}^{2n}\to\mathbb{R} $ is PD.

Using $ V^{*} $ as a candidate Lyapunov function and using the HJB equation in \eqref{eq:HJB}, it can be concluded that\begin{align}
V^{*}_{s_{1}}\left(s\right)H(s)+V^{*}_{s_{2}}\left(s\right)\left(F\left(s\right)+G\left(s\right)u^{*}\left(s\right)\right)\leq-Q\left(s\right),
\end{align}$ \forall s\in\mathbb{R}^{2n}. $ If Assumption \ref{ass:CostRestrictions}-(a) holds, then the proof is complete using Lyapunov's direct method. If Assumption \ref{ass:CostRestrictions}-(b) holds, then using the fact that if the output is identically zero then so is the state, the invariance principle \cite[Corollary 4.2]{SCC.Khalil2002} can be invoked to complete the proof. If Assumption \ref{ass:CostRestrictions}-(c) holds, then finiteness of the value function everywhere implies that the origin is the only equilibrium point of the closed-loop system. As a result, the invariance principle can be invoked to complete the proof.
\end{proof}

\end{document}